\documentclass[english,final]{IEEEtran}
 
\makeatletter

\usepackage[T1]{fontenc}
\usepackage[latin9]{inputenc}
\usepackage{amsthm}
\usepackage{amsmath}
\usepackage{amssymb}
\usepackage{graphicx}
\usepackage{dsfont}
\usepackage{cite}
\usepackage{amsmath}
\usepackage{array}
\usepackage{multirow}
\usepackage{caption}
\usepackage{color}

\theoremstyle{plain}

\theoremstyle{plain}

\theoremstyle{plain}
\newtheorem{lem}{\protect\lemmaname}
\theoremstyle{plain}
\newtheorem{thm}{\protect\theoremname}
\theoremstyle{plain}
  
\theoremstyle{definition}

\theoremstyle{definition}

\theoremstyle{definition}
\newtheorem{rem}{\protect\remarkname}

\makeatother
  
\usepackage{babel} 

\providecommand{\claimname}{Claim}
\providecommand{\lemmaname}{Lemma}
\providecommand{\propositionname}{Proposition}
\providecommand{\theoremname}{Theorem}
\providecommand{\corollaryname}{Corollary} 
\providecommand{\definitionname}{Definition}
\providecommand{\assumptionname}{Assumption}
\providecommand{\remarkname}{Remark}

\newcommand{\overbar}[1]{\mkern 1.25mu\overline{\mkern-1.25mu#1\mkern-0.25mu}\mkern 0.25mu}

\DeclareMathOperator*{\argmax}{arg\,max}

\newcommand{\openone}{\mathds{1}}

\newcommand{\lammin}{\lambda_{\mathrm{min}}}
\newcommand{\lammax}{\lambda_{\mathrm{max}}}

\newcommand{\Gbar}{\overbar{G}}

\newcommand{\qmax}{q_{\mathrm{max}}}

\newcommand{\dmax}{d_{\mathrm{max}}}

\newcommand{\pebar}{\overbar{P}_{\mathrm{e}}}

\newcommand{\pe}{P_{\mathrm{e}}}

\newcommand{\Xv}{\mathbf{X}}

\newcommand{\Ec}{\mathcal{E}}

\newcommand{\Gc}{\mathcal{G}}

\newcommand{\Tc}{\mathcal{T}}

\newcommand{\EE}{\mathbb{E}}
\newcommand{\PP}{\mathbb{P}}

\newcommand{\mbar}{\overbar{m}}

\usepackage{float}

\usepackage{color}
\usepackage{array}
\usepackage{hyperref}

\begin{document}

\title{On the Difficulty of Selecting Ising \\ Models with Approximate Recovery}

\author{Jonathan Scarlett and Volkan Cevher}
\maketitle

\begin{abstract}
    In this paper, we consider the problem of estimating the underlying graph associated with an Ising model given a number of independent and identically distributed samples.  We adopt an \emph{approximate recovery} criterion that allows for a number of missed edges or incorrectly-included edges, in contrast with the widely-studied exact recovery problem.  Our main results provide information-theoretic lower bounds on the sample complexity for graph classes imposing constraints on the number of edges, maximal degree, and other properties.  We identify a broad range of scenarios where, either up to constant factors or logarithmic factors, our lower bounds match the best known lower bounds for the exact recovery criterion, several of which are known to be tight or near-tight.  Hence, in these cases, approximate recovery has a similar difficulty to exact recovery in the minimax sense.
    
    Our bounds are obtained via a modification of Fano's inequality for handling the approximate recovery criterion, along with suitably-designed ensembles of graphs that can broadly be classed into two categories: (i) Those containing graphs that contain several isolated edges or cliques and are thus difficult to distinguish from the empty graph; (ii) Those containing graphs for which certain groups of nodes are highly correlated, thus making it difficult to determine precisely which edges connect them.  We support our theoretical results on these ensembles with numerical experiments.
\end{abstract}

\long\def\symbolfootnote[#1]#2{\begingroup\def\thefootnote{\fnsymbol{footnote}}\footnote[#1]{#2}\endgroup}

\symbolfootnote[0]{ The authors are with the Laboratory for Information and Inference Systems (LIONS), \'Ecole Polytechnique F\'ed\'erale de Lausanne (EPFL) (e-mail: \{jonathan.scarlett,volkan.cevher\}@epfl.ch).

This work was supported in part by the European Commission under Grant ERC Future Proof, SNF 200021-146750 and SNF CRSII2-147633, and EPFL Fellows Horizon2020 grant 665667.}
\vspace*{-0.5cm}

\section{Introduction}

Graphical models are a widely-used tool for providing compact representations of the conditional independence relations between random variables, and arise in areas such as image processing \cite{Gem84}, statistical physics \cite{Gla63}, computational biology \cite{Dur98}, natural language processing \cite{Man99}, and social network analysis \cite{Was94}.  The problem of \emph{graphical model selection} consists of recovering the graph structure given a number of independent samples from the underlying distribution.  

While this fundamental problem is NP-hard in general \cite{Chi96}, there exist a variety of methods guaranteeing \emph{exact} recovery with high probability on \emph{restricted} classes of graphs, such as bounded degree and bounded number of edges.  Existing works have focused primarily on Ising models and Gaussian models, and our focus in this paper is on the former.

In particular, we focus in the problem of \emph{approximate recovery}, in which one can tolerate some number of missed edges or incorrectly-included edges.  The motivation for such a study is that the exact recovery criterion is very restrictive, and not something that one would typically expect to achieve in practice.  In particular, if the number of samples required for exact recovery is very large, it is of significant interest to know the potential savings by allowing for approximate recovery.  The answer is unclear \emph{a priori}, since this can lead to vastly improved scaling laws in some inference and learning problems \cite{Ree12} and virtually no gain in others \cite{Sca15b}.


Our main focus is on algorithm-independent lower bounds for Ising models, revealing the number of measurements required for approximate recovery regardless of the computational complexity.  We extend Fano's inequality \cite{San12,Sha14} to the case of approximate recovery, and apply it to restricted sets of graphs that prove the difficulty of approximate recovery.  

Our main results reveal a broad range of graph classes for which the approximate recovery lower bounds exhibit the same scalings as the best-known exact recovery lower bounds \cite{San12,Sha14}, which are known to be tight or near-tight in many cases of interest.  This indicates that, at least for the classes that we consider, the approximate recovery problem is not much easier than the exact recovery problem in the minimax sense.  

\subsection{Problem Statement} \label{sec:SETUP}

The ferromagnetic Ising model \cite{Isi25} is specified by a graph $G = (V,E)$ with vertex set $V = \{1,\dotsc,p\}$ and edge set $E$.  Each vertex is associated with a binary random variable $X_i \in \{-1,1\}$, and the corresponding joint distribution is
\begin{equation}
    P_{G}(x) = \frac{1}{Z}\exp\bigg( \sum_{i,j} \lambda_{ij} x_i x_j \bigg), \label{eq:Ising}
\end{equation}
where
\begin{equation}
    \lambda_{ij} = \begin{cases} \lambda & (i,j) \in E \\ 0 & \mathrm{otherwise}, \end{cases}     \label{eq:lambda}
\end{equation}
and $Z$ is a normalizing constant called the partition function. Here $\lambda > 0$ is a parameter to the distribution, sometimes called the inverse temperature.

Let $\Xv \in \{0,1\}^{n \times p}$ be a matrix of $n$ independent samples from this distribution, each row corresponding to one such sample of the $p$ variables. Given $\Xv$, an \emph{estimator} or \emph{decoder} constructs an estimate $\hat{G}$ of the graph $G$, or equivalently, an estimate $\hat{E}$ of the edge set $E$.

\textbf{Recovery Criterion:} Given some class $\Gc$ of graphs, the widely-studied exact recovery criterion seeks to characterize
\begin{equation}
    \pe := \max_{G \in \Gc} \PP[\hat{E} \ne E]. \label{eq:pe}
\end{equation}
We instead consider the following approximate recovery criterion, for some maximum number of errors $\qmax \ge 0$:
\begin{equation}
    \pe(\qmax) := \max_{G \in \Gc} \PP\big[|E \Delta \hat{E}| > \qmax\big],  \label{eq:pe_approx}
\end{equation}
where $E \Delta \hat{E} = (E \backslash \hat{E}) \cup (\hat{E} \backslash E)$, so that $|E \Delta \hat{E}|$ denotes the \emph{edit distance}, i.e., the number of edge insertions and deletions required to transform one graph to another.  In this definition, $\qmax$ does not depend on $G$, and hence, the number of allowed edge errors does not depend on the graph itself.  We consider graph classes with a maximum number of edges equal to some value $k$, and set $\qmax = \theta^{*} k$ for some constant $\theta^* \in (0,1)$ not scaling with the problem size.  Note that $\theta^* = 1$ would trivially give $\pe(\qmax) = 0$.

\textbf{Graph Classes:} We consider the following three nested classes of graphs $\Gc_k \supseteq \Gc_{k,d} \supseteq \Gc_{k,d,\eta,\gamma}$:
\begin{itemize}
    \item \emph{(Edge bounded class $\Gc_k$)} This class contains all graphs with at most $k$ edges.
    \item \emph{(Edge and degree bounded class $\Gc_{k,d}$)} This class contains the graphs in $\Gc_k$ such that each node has degree (i.e., number of edges it is involved in) at most $d$.
    \item \emph{(Sparse separator class $\Gc_{k,d,\eta,\gamma}$)} This class contains the graphs in $\Gc_{k,d}$ satisfying the \emph{$(\eta,\gamma)$-separation condition} \cite{Ana12}: For any two non-connected vertices in the graph, one can simultaneously block all paths of length $\gamma$ or less by blocking at most $\eta$ nodes.
\end{itemize}
The restriction on the number of edges is motivated by the fact that real-world graphs are often sparse.  The restriction on the degree is also relevant in applications, and is particularly commonly-assumed in the statistical physics literature.  The sparse separation condition is somewhat more technical, but it is of interest since it is known to permit polynomial-time exact recovery in many cases \cite{Ana12,Ana12a}.  Moreover, it is known to hold with high probability for several interesting random graphs; see \cite{Ana12} for some examples.

\textbf{Generalized Edge Weights:} A generalization of the above Ising model allows $\lambda_{ij}$ to take different non-zero values for each $(i,j) \in E$, some of which may be negative.  Previous works considering model selection for this generalized model have sought minimax bounds with respect to the graph class \emph{and} these parameters subject to $\lammin \le |\lambda_{ij}| \le \lammax$ for some $\lammin$ and $\lammax$.  The lower bounds derived in this paper immediately imply corresponding lower bounds for this generalized setting, provided that our parameter $\lambda$ in \eqref{eq:lambda} lies in the range $[\lammin,\lammax]$.

\textbf{Notation and Terminology:} Throughout the paper, we let $\PP_G$ and $\EE_G$ denote probabilities and expectations with respect to $P_G$ (e.g., $\PP_G[X_i = X_j]$, $\EE[X_iX_j]$).  We denote the floor function by $\lfloor\cdot\rfloor$, and the ceiling function by $\lceil\cdot\rceil$.  We use the standard terminology that the \emph{degree} of a node $v \in V$ is the number of edges in $E$ containing $v$, and that a \emph{clique} is a subset $C \subset V$ of size at least two within which all pairs of nodes have an edge between them.

\subsection{Related Work}

A variety of algorithms with varying levels of computational efficiency have been proposed for selecting Ising models with rigorous guarantees, including conditional independence tests for candidate neighborhoods \cite{Bre08}, correlation tests in the presence of sparse separators \cite{Ana12,Wu13}, greedy techniques \cite{Jal11,Ray12,Bre14,Bre14a}, convex optimization approaches \cite{Rav10}, elementary estimators \cite{Yan14}, and intractable information-theoretic techniques \cite{San12}.  

These works have made various assumptions on the underlying model, including incoherence assumptions \cite{Rav10,Yan14} and long-range correlation assumptions \cite{Ana12,Wu13}.  A notable recent work avoiding these is \cite{Bre14a}, which provides recovery guarantees using an algorithm whose complexity is only quadratic in the number of nodes for a fixed maximum degree, thus resolving an open question posed in \cite{Mon09}.

Early works providing algorithm-independent lower bounds used only graph-theoretic properties \cite{Bre08,Ana12,Tan13}; the resulting bounds are loose in general, since they do not capture the effects of the parameters of the joint distribution (e.g., $\lambda$).  Several refined bounds were given in \cite{San12} for graphs with a bounded degree or a bounded number of edges.  Additional classes were considered in \cite{Sha14}, including the bounded girth class and a class related to the separation criterion of \cite{Ana12} (and hence related to $\Gc_{k,d,\eta,\gamma}$ defined above).  While our techniques build on those of \cite{San12,Sha14}, we must consider significantly different ensembles, since those in \cite{San12,Sha14} contain graphs that differ only by one or two edges, thus making approximate recovery trivial. 

To our knowledge, the only other work giving an approximate recovery bound for the Ising model is \cite{Das12}, where the degree-bounded class is considered.  The effect of edge weights is not considered therein, and the bound is proved by counting graphs rather than constructing restricted ensembles.  Consequently, only an $\Omega(d \log p)$ necessary condition is shown, in contrast with our bounds containing a $d^2$ or $e^{\lambda d}$ term (\emph{cf.}, Table \ref{tbl:summary}).  The necessary conditions for list decoding \cite{Vat11} bear some similarity to approximate recovery, but the problem and its analysis are in fact much more similar to exact recovery, allowing the ensembles from \cite{San12,Sha14} to be applied directly.
 
Beyond Ising models, several works have provided necessary and sufficient conditions for recovering Gaussian graphical models \cite{Mei06,Wan10,Ana12a,Jog15,Rav11}.  In this context, a necessary condition for approximate recovery was given in \cite[Cor.~7]{Ana12a}, but the corresponding assumptions and techniques used were vastly different to ours: The random Erd\"os-R\'enyi model was considered instead of a deterministic class, and an additional walk-summability condition specific to the Gaussian model was imposed. 

\subsection{Contributions} \label{sec:CONTIRBUTIONS}

Our main results, and the corresponding existing results for exact recovery, are summarized in Table \ref{tbl:summary}, where we provide necessary scaling laws on the number of samples needed to obtain a vanishing probability of error $\pe(\qmax)$.  Note that some of the exact recovery conditions given in the final column were not explicitly given in \cite{San12,Sha14}, but they can easily be inferred from the proofs therein; see Section \ref{sec:RESULTS} for further discussion.  We also observe that our analysis requires handling more cases separately compared to \cite{San12,Sha14}; in those works, the final three rows corresponding to $\Gc_k$ in Table \ref{tbl:summary} are all a single case giving $\Omega(k \log p)$ scaling, and similarly for $\Gc_{k,d}$.

\begin{table*}
    \begin{centering}
    \begin{tabular}{|>{\centering}m{4.2cm}|>{\centering}m{3.6cm}|>{\centering}m{3.8cm}|>{\centering}m{3.2cm}|}
    \hline 
    Graph Class  & Parameters  & Necessary for approximate recovery (this paper) & Best known necessary for exact recovery \cite{San12,Sha14}\tabularnewline
    \hline 
    \hline 
    \multirow{4}{4.2cm}{ {\centering \\ ~ \\
    \textbf{Bounded edge $\mathcal{G}_{k}$} \\
    \smallskip
    Distortion $q_{\mathrm{max}} < \frac{k}{4}$ \\
    \smallskip
    (Theorems~\ref{thm:necc_k1} and \ref{thm:necc_k2}) \\ }} & $\lambda=\omega\big(\frac{1}{\sqrt{k}}\big)$  & Exponential in $\lambda\sqrt{k}$ & Exponential in $\lambda\sqrt{k}$\tabularnewline
    \cline{2-4} 
     & $\lambda=O\big(\frac{1}{\sqrt{k}}\big)$
    
    $1 \ll k\ll p$ & $\Omega(k\log p)$ & $\Omega(k\log p)$\tabularnewline
    \cline{2-4} 
     & $\lambda=O\big(\frac{1}{\sqrt{k}}\big)$
    
    $p\ll k\ll p^{\frac{4}{3}}$ & $\Omega(k)$ & $\Omega(k\log p)$\tabularnewline
    \cline{2-4} 
     & $\lambda=O\big(\frac{1}{\sqrt{k}}\big)$
    
    $p^{\frac{4}{3}}\ll k\ll p^{2}$ & $\Omega\Big(\frac{p^{2}}{\sqrt{k}}\Big)$
    
    {\scriptsize (between $\Omega(p)$ and $\Omega(k)$)} & $\Omega(k\log p)$\tabularnewline
    \hline 
    \multirow{4}{4.2cm}{{\centering \\ ~ \vspace*{-0.2cm} \\  \textbf{Bounded edge and degree $\mathcal{G}_{k,d}$} \\
    \medskip
    Distortion $q_{\mathrm{max}} < \frac{k}{4}\frac{d-2}{d}$ \\
    \medskip
    (Theorems~\ref{thm:necc_d1} and \ref{thm:necc_d2}) \\ }} & $\lambda=\omega\big(\frac{1}{d}\big)$  & Exponential in $\lambda d$ & Exponential in $\lambda d$\tabularnewline
    \cline{2-4} 
     & $\lambda=O\big(\frac{1}{d}\big)$
    
    $d^2 \ll k \ll p$ & $\Omega(d^{2}\log p)$ & $\Omega(d^{2}\log p)$\tabularnewline
    \cline{2-4} 
     & $\lambda=O\big(\frac{1}{d}\big)$
    
    $p\ll k\ll p\sqrt{d}$ & $\Omega(d^{2})$ & $\Omega(d^{2}\log p)$\tabularnewline
    \cline{2-4} 
     & $\lambda=O\big(\frac{1}{d}\big)$
    
    $p\sqrt{d}\ll k\le\frac{pd}{2}$ & $\Omega\Big(\frac{d^{3}p^{2}}{k^{2}}\Big)$
    
    {\scriptsize (between $\Omega(d)$ and $\Omega(d^{2})$)} & $\Omega(d^{2}\log p)$\tabularnewline
    \hline 
    \multirow{2}{4.2cm}{{\centering \vspace*{-0.63cm} \\ \textbf{Bounded edge and degree with sparse separators $\mathcal{G}_{k,d,\eta,\gamma}$} \\
    \medskip
    Distortion $q_{\mathrm{max}}<\frac{(c\eta-1)^{2}k}{2c\eta(2\eta+m(\gamma+1))}$ 
    ($c\in(0,1)$, $m\in\{0,\dotsc,\frac{d}{2}-\eta\}$) \\
    \medskip
    (Theorem~\ref{thm:necc_s1}) \\ }} & \smallskip $\lambda=\omega\Big(\min\Big\{\frac{1}{\sqrt{\eta}},\frac{1}{m^{\frac{1}{1+\gamma}}}\Big\}\Big)$
    
    $\lambda = O(1)$
    
    $k\le\frac{p}{4}$ \smallskip &  Exponential in $\max\big\{\lambda^2\eta,\lambda^{\gamma+1}m\big\}$  & Exponential in $\max\big\{\lambda^2\eta,\lambda^{\gamma+1}d\big\}$\tabularnewline
    \cline{2-4} 
     & \smallskip $\lambda=O\Big(\min\Big\{\frac{1}{\sqrt{\eta}},\frac{1}{m^{\frac{1}{1+\gamma}}}\Big\}\Big)$
     
     $\lambda = O(1)$
    
    $k\le\frac{p}{4}$ \smallskip & $\Omega\Big(\max\Big\{\eta,m^{\frac{2}{\gamma+1}}\Big\}\log p\Big)$ & $\Omega\Big(\max\Big\{\eta,d^{\frac{2}{\gamma+1}}\Big\}\log p\Big)$\tabularnewline
    \hline 
    \end{tabular}
    \par\end{centering}

\protect\protect\caption{Summary of main results on parital recovery, and comparisons to the
best known necessary conditions for exact recovery.  Each entry shows the necessary scaling law for the number of samples required to achieve a vanishing error probability. \label{tbl:summary}} \vspace*{-3ex}
\end{table*}

 Table \ref{tbl:summary} reveals the following facts:
\begin{enumerate}
    \item In all of the known cases where exact recovery is known to be difficult, i.e., exponential in a quantity that increases in the problem dimension, the same difficulty is observed for approximate recovery, at least for the values of $\qmax$ shown.  For $\Gc_k$ and $\Gc_{k,d}$, this is true even when we allow for up to a quarter of the edges to be in error.  Note that we did not seek to optimize this fraction in our analysis, and we expect similar difficulties to arise even when higher proportions of errors are allowed.  In fact, by a simple variation of our analysis outlined in Remark \ref{rem:extension} in Section \ref{sec:ENS3}, we can already increase this fraction from $\frac{1}{4}$ to $\frac{1}{2}$.
    \item In many of the cases where the necessary conditions for exact recovery lack exponential terms, the corresponding necessary conditions for approximate recovery are identical or near-identical; in particular, see the second and third rows corresponding to $\Gc_k$, the second and third rows corresponding to $\Gc_{k,d}$, and the second row corresponding to $\Gc_{k,d,\eta,\gamma}$ with $m = \frac{d}{2} - \eta$.  While there are logarithmic terms missing in some cases (e.g., $k$ vs.~$k\log p$), these are typically insignificant in the regimes considered (e.g., $k=\Omega(p)$).
    \item In contrast, there are some cases where significant gaps remain between the best-known conditions for exact recovery and approximate recovery.  The two most extreme cases are as follows: (i) If $k=\Theta(p^{2-\epsilon})$ for some small $\epsilon > 0$, the necessary conditions for $\Gc_k$ are $\Omega(p^{2-\epsilon} \log p)$ and $\Omega(p^{1+\epsilon/2})$, respectively; (ii) If $k = \Theta(pd)$, then the necessary conditions for $\Gc_{k,d}$ are $\Omega(d^2 \log p)$ and $\Omega(d \log p)$, respectively.  It remains an open problem as to whether this behavior is fundamental, or due to a weakness in the analysis.
\end{enumerate}

The starting point of our results is a modification of Fano's inequality for the purpose of handling approximate recovery.  To obtain the above results, we apply this bound to ensembles of graphs that can be broadly classed into two categories.  The first considers graphs with a large number of isolated edges, or more generally, isolated cliques.  We characterize how difficult each graph is to distinguish from the empty graph, and use this to derive the results given in item 2) above.  On the other hand, the results on the exponential terms discussed in item 1) arise from considering ensembles in which several groups of nodes are always highly correlated due to the presence of a large number of edges among them, thus making it difficult to determine precisely which edges these are.

Both of these categories help in providing bounds that match those for exact recovery.  For example, the $\Omega(k \log p)$ behavior for $\lambda = O\big(\frac{1}{k}\big)$ in \cite{San12} is proved by considering graphs with a single isolated edge, and our analysis extends this to approximate recovery by considering graphs with $k$ isolated edges.  Analogously, the exponential behavior (e.g., in $\lambda\sqrt{k}$) in \cite{San12} is proved by considering cliques with one edge removed, and our analysis reveals that the same exponential behavior arises even if a constant fraction of the the edges are removed.

We provide numerical results on our ensembles in Section \ref{sec:NUMERICAL} supporting our theoretical findings.  Specifically, we implement optimal or near-optimal decoding rules in a variety of cases, and find that while approximate recovery can be easier than exact recovery, the general behavior of the two is similar.

\section{Main Results} \label{sec:RESULTS}

In this section, we present our main results, namely, algorithm-independent necessary conditions for the criterion in \eqref{eq:pe_approx} with all $\lambda_{ij} = \lambda$.  Our conditions are written in terms of asymptotic $o(1)$ terms for clarity, but purely non-asymptotic variants can be inferred from the proofs.  Throughout the section, we make use of the binary entropy function in nats, $H_2(\theta) := -\theta \log \theta - (1-\theta)\log(1-\theta)$.  Here and subsequently, all logarithms have base $e$.

All proofs are deferred to later sections; some preliminary results are presented in Section \ref{sec:AUXILIARY}, a number of ensembles are presented and analyzed in Section \ref{sec:ENSEMBLES}, and the resulting theorems are deduced in Section \ref{sec:APPLICATIONS}.

\subsection{Bounded Number of Edges Class $\Gc_k$}

We first consider the class $\Gc_{k}$ of graphs with at most $k$ edges.  It will prove convenient to treat two cases separately depending on how $k$ scales with $p$.

\begin{thm} \label{thm:necc_k1}
    \emph{(Class $\Gc_{k}$ with $k \le p/4$)} For any number of edges such that $k \to \infty$ and $k \le p/4$, and any distortion level $\qmax = \lfloor \theta k \rfloor$ for some $\theta \in \big(0,\frac{1}{4}\big)$, it is necessary that
    \begin{multline}
        n \ge \max\bigg\{ \frac{ e^{\lambda(\sqrt{k/2}-2)/2} \big( \log 2  - H_2(2\theta) \big) }{ 6 \lambda k }, \\ \frac{ 2(1-\theta)\log p }{ \lambda\tanh\lambda }\bigg\} \Big( 1-\delta - o(1) \Big) \label{eq:k_cond_final1}
    \end{multline}
    in order to have $\pe(\qmax) \le \delta$ for all $G \in \Gc_k$.
\end{thm}

We proceed by considering two cases as in \cite{San12}.  In the case that $\lambda\sqrt{k} \to \infty$ at any rate faster than logarithmic in $p$ (or even logarithmic with a constant that is not too small), the sample complexity is dominated by the exponential term $e^{\lambda(\sqrt{k}-2)/2}$, and any recovery procedure requires a huge number of samples.  Thus, in this case, even the approximate recovery problem is very difficult.  On the other hand, if $\lambda = O\big(\frac{1}{\sqrt k}\big)$ then the second condition in \eqref{eq:k_cond_final1} gives a sample complexity of $\Omega(k \log p)$, since $\tanh \lambda = O(\lambda)$ as $\lambda \to 0$.  

These observations are the same as those made for exact recovery in \cite{San12}, where the best known necessary conditions for $\Gc_k$ were given. Thus, we have reached similar conclusions even allowing for nearly a quarter of the edges to be in error.

\begin{thm} \label{thm:necc_k2}
    \emph{(Class $\Gc_{k}$ with $k=\Omega(p)$)} For any number of edges of the form $k = \lfloor c p^{1+\nu} \rfloor$ for constants $c>0$ and $\nu\in[0,1)$, and any distortion level $\qmax = \lfloor \theta k \rfloor$ for some $\theta \in \big(0,\frac{1}{4}\big)$, it is necessary that
    \begin{multline}
        n \ge \max\bigg\{ \frac{ e^{\lambda(\sqrt{k/2}-2)/2} \big(\log 2  - H_2(2\theta)\big) }{ 6 \lambda k }, \\ \frac{ \log 2  - H_2(\theta)  }{ \lambda \frac{e^{2\lambda}\cosh(4\lambda cp^{\nu}) - 1}{e^{2\lambda}\cosh(4\lambda cp^{\nu}) + 1} } \bigg\}  \Big( 1-\delta - o(1) \Big)  \label{eq:k_cond_final2}
    \end{multline}
    in order to have $\pe(\qmax) \le \delta$ for all $G \in \Gc_k$.
\end{thm}

As above, the sample complexity is exponential in $\lambda\sqrt{k}$ due to the first term in \eqref{eq:k_cond_final2}.  On the other hand, we claim that when $\lambda = O\big(\frac{1}{\sqrt k}\big)$, the second term in \eqref{eq:k_cond_final2} leads to the sample complexity $O(\min\{k,p^2/\sqrt{k}\})$.  To see this, we choose $k$ as in the theorem statement and note that $\lambda p^{\nu} = O( p^{-\frac{1}{2}(1+\nu) + \nu } ) = O( p^{-\frac{1}{2}(1-\nu) } )$; since $\cosh\zeta = 1 + O(\zeta^2)$ as $\zeta \to 0$, this  implies that $\cosh(2c\lambda p^{\nu}) = 1 + O(p^{-(1-\nu)})$.  We thus have $e^{2\lambda}\cosh(2c\lambda p^{\nu}) = \big(1 + O(p^{-\frac{1}{2}(1+\nu)})\big)\big((1 + O(p^{-(1-\nu)})\big)$, which finally yields $\frac{e^{2\lambda}\cosh(2c\lambda p^{\nu}) - 1}{e^{2\lambda}\cosh(2c\lambda p^{\nu}) + 1} = O(\max\{ p^{-\frac{1}{2}(1+\nu)}, p^{-(1-\nu)} \}) = O(\max\{1/\sqrt{k}, k/p^2\})$.

When $k = \Omega(p)$ and $k = O(p^{4/3})$, we have $\min\{k,p^2/\sqrt{k}\} = k$, and hence, these observations are again the same as those made for exact recovery in \cite{San12}, except that our growth rates do not include a $\log p$ term; this logarithmic factor is insignificant compared to the leading term $k=\Omega(p)$.  In contrast, the gap is more significant when $k \gg p^{4/3}$; in the extreme case, when $k=\Theta(p^{2-\epsilon})$ for some small $\epsilon > 0$, we obtain a scaling of $\Omega(p^{1+\epsilon/2})$, as opposed to $\Omega(k \log p) = \Omega(p^{2-\epsilon} \log p)$.

\subsection{Bounded Degree Class $\Gc_{k,d}$}

Next, we consider the glass $\Gc_{k,d}$ of graphs such that every node has degree at most $d$, and the total number of edges does not exceed $k$.  

\begin{thm} \label{thm:necc_d1}
    \emph{(Class $\Gc_{k,d}$ with $k \le p/4$)} For any maximal degree $d > 2$ and number of edges $k$ such that $k = \omega(d^2)$ and  $k \le p/4$, and any distortion level $\qmax = \lfloor \theta k \rfloor$ for some $\theta \in \big(0,\frac{1}{4}\frac{d-2}{d}\big)$, it is necessary that
    \begin{multline}
        n \ge \max\bigg\{ \frac{ e^{\lambda (d-2)/4} \big( \log 2  - H_2\big(\frac{d}{d - 2} \cdot 2\theta\big) \big) }{ 3 \lambda d^2 } , \\ \frac{ 2(1-\theta)\log p }{ \lambda\tanh\lambda } \bigg\} \Big( 1-\delta - o(1) \Big) \label{eq:d_cond_final1}
    \end{multline}
    in order to have $\pe(\qmax) \le \delta$ for all $G \in \Gc_{k,d}$.
\end{thm}

The first term in \eqref{eq:d_cond_final1} reveals that the sample complexity is exponential in $\lambda d$.  On the other hand, if $\lambda = O\big(\frac{1}{d}\big)$ then the second term gives a sample complexity of $\Omega(d^2 \log p)$. 

We cannot directly compare Theorem \ref{thm:necc_d1} to \cite{San12}, since there $k$ was assumed to be unrestricted for the degree-bounded ensemble.  However, the analysis therein is easily extended to $\Gc_{k,d}$, and doing so recovers the nearly identical observations to those above, as summarized in Table \ref{tbl:summary}.  In this sense, Theorem \ref{thm:necc_d1} matches the best known necessary conditions for exact recovery even when nearly a quarter of the edges may be in error.

\begin{thm} \label{thm:necc_d2}
    \emph{(Class $\Gc_{k,d}$ with $k=\Omega(p)$)} For any maximal degree $d > 2$ and number of edges $k$ such that $k = \omega(d^2)$ and $k \le \frac{1}{2} p (d'-1)$ for some $d' \le d$, and any distortion level $\qmax = \lfloor \theta k \rfloor$ for some $\theta \in \big(0,\frac{1}{4}\frac{d-2}{d}\big)$, it is necessary that
    \begin{multline}
        n \ge \max\bigg\{ \frac{ e^{\lambda (d-2)/4} \big( \log 2  - H_2\big(\frac{d}{d - 2} \cdot 2\theta\big) \big) }{ 3 \lambda d^2 }, \\ \frac{ \log 2  - H_2(\theta) }{ \lambda \frac{e^{2\lambda}\cosh(2\lambda d') - 1}{e^{2\lambda}\cosh(2\lambda d') + 1} } \bigg\} \Big( 1-\delta - o(1) \Big) \label{eq:d_cond_final2}
    \end{multline}
    in order to have $\pe(\qmax) \le \delta$ for all $G \in \Gc_{k,d}$.
\end{thm}

The sample complexity remains exponential in $\lambda d$. By some standard asymptotic expansions similar to those following Theorem \ref{thm:necc_k2}, we have $\frac{e^{2\lambda}\cosh(2\lambda d') - 1}{e^{2\lambda}\cosh(2\lambda d') + 1} = O\big( \max\big\{ \frac{1}{d}, \big( \frac{d'}{d} \big)^2 \big\}\big)$ whenever $\lambda = O\big(\frac{1}{d}\big)$; hence, the second condition in \eqref{eq:d_cond_final2} becomes $n = \Omega\big( d \min\big\{d, \big( \frac{d}{d'} \big)^2 \big\} \big)$.  Thus, if $d' = O(\sqrt{d})$ then we again get the desired $n = O(d^2 \log p)$  behavior; this means that we can allow for $k$ up to $O(p\sqrt{d})$.  More generally, we instead get the possibly weaker scaling law $n = \Omega\big( \min\big\{ d^2, d^3/(d')^2 \big\} \big)$, which is equivalent to $n = \Omega\big( \min\big\{ d^2,  \frac{d^3 p^2}{ k^2 } \big\} \big)$ when $k=\Theta(pd')$.  In the extreme case, when $k = \Theta(pd)$ (the highest growth rate possible given the degree constraint alone), this only recovers $\Omega(d \log p)$ scaling.

\subsection{Sparse Separator Class $\Gc_{k,d,\eta,\gamma}$} \label{sec:SPARSE_ENS}

We now consider the class $\Gc_{k,d,\eta,\gamma}$ of graphs in $\Gc_{k,d}$ that satisfy the $(\eta,\gamma)$-separation condition \cite{Ana12}.  We focus on the case $k \le p/4$, since the main graph ensemble that we consider for this class is not suited to the case that $k = \omega(p)$.

\begin{thm} \label{thm:necc_s1}
    \emph{(Class $\Gc_{k,d,\eta,\gamma}$ with $k \le p/4$)} Fix any parameters $(d,k,\eta,\gamma)$ with $k \le p/4$ and $\eta \le \lfloor\frac{d}{2}\rfloor$, and let $m$ be an integer in $\big\{0,\dotsc,\lfloor\frac{d}{2}\rfloor -\eta\big\}$.  For any distortion level $\qmax = \big\lfloor \theta \frac{(c\eta-1)^2 k}{ 2c\eta( 2\eta + m(\gamma+1) ) } \big\rfloor$ for some $\theta \in \big(0,\frac{1}{2}\big)$ and $c \in \big(\frac{1}{\eta},1\big]$, it is necessary that
    \begin{multline}
        n \ge \max\bigg\{ \frac{ \Big(1 + \big(\cosh(2\lambda)\big)^{(1-c)\eta - 1} \big( \frac{ 1+(\tanh\lambda)^{\gamma+1}}{1-(\tanh\lambda)^{\gamma+1}} \big)^{m} \Big) }{ 2\lambda c \eta }  \\ \times\big( \log 2  - H_2(\theta) \big), \frac{2(k-\qmax)\log p }{ k\lambda\tanh\lambda }\bigg\} \Big( 1-\delta - o(1) \Big) \label{eq:s_cond_final}
    \end{multline}
    in order to have $\pe(\qmax) \le \delta$ for all $G \in \Gc_{k,d,\eta,\gamma}$.
\end{thm}

We proceed by considering only the case $\lambda = O(1)$, though simplifications of Theorem \ref{thm:necc_s1} for $\lambda \to \infty$ are also possible.  With $\lambda = O(1)$, we have $\big(\cosh(2\lambda)\big)^{(1-c)\eta} = e^{ \zeta \lambda^2(1-c)\eta}$ for some $\zeta = \Theta(1)$, and similarly $\big( \frac{1+(\tanh\lambda)^{\gamma+1}}{1-(\tanh\lambda)^{\gamma+1}} \big)^{m} = e^{\zeta' m\lambda^{\gamma+1}}$ for some $\zeta' = \Theta(1)$ \cite[Sec.~5]{Sha14}.  These identities reveal that the sample complexity is exponential in both $\lambda^2\eta$ and $\lambda^{\gamma+1}m$. On the other hand, if $\lambda = O\big( \frac{1}{\sqrt{\eta}} \big)$ and $\lambda = O\big( \frac{1}{m^{\frac{1}{\gamma+1}}} \big)$ then the second term in \eqref{eq:s_cond_final} gives $n = \Omega(\max\{\eta,m^{\frac{2}{\gamma+1}}\} \log p)$.

Due to the choice $\qmax = \big\lfloor \theta \frac{(c\eta-1)^2 k}{ 2c\eta( 2\eta + m(\gamma+1) ) } \big\rfloor$, if we set $m = d/2 - \eta$ then we are only in the regime of a constant fraction of errors if $d\gamma = \Theta(\eta)$.  This is true, for example, if $\eta = \Theta(d)$  so that the separator set size is a fixed fraction of the maximum degree, and $\gamma = \Theta(1)$ so that the separation is with respect to paths of a bounded length.  

More generally, to handle larger values of $\qmax$, one can choose a smaller value of $m$, thus leading to a larger value of $\qmax$ but with a less stringent condition on the number of measurements in \eqref{eq:s_cond_final}.  In the extreme case, $m=0$, and then we are always in the regime of a constant proportion of errors; however, this yields a necessary condition $\Omega(\eta\log p)$ not depending on $d$ or $\gamma$.

The graph family studied in \cite[Thm.~2]{Sha14} is somewhat different from $\Gc_{k,d,\eta,\gamma}$, in particular not putting any constraints on the maximal degree nor the number of edges.  Nevertheless, by choosing the parameters in the proof therein to meet these constraints,\footnote{Specifically, in \cite[Sec.~9.2]{Sha14}, one can set $t_{\nu} = d - \eta$ to satisfy the degree constraint, and then choose $\alpha = \big\lfloor \frac{k}{t_{\nu}(\gamma+1) + 2\eta - 1} \big\rfloor$ to ensure there are at most $k$ edges in total.} one again obtains similar conditions to those above, as summarized in Table \ref{tbl:summary}.  In particular, for any choice of $m$ that grows as $\Theta(d)$, the scaling laws for exact recovery and approximate recovery coincide. 

\section{Auxiliary Results} \label{sec:AUXILIARY}

In this section, we provide a number of auxiliary results that will be used to prove the theorems in Section \ref{sec:RESULTS}.  We first present a general form of Fano's depending on both the Kullback-Leibler (KL) divergence and edit distance between graphs, and then provide a number of properties of Ising models that will be useful for characterizing the KL divergence and edit distance in specific scenarios.

\subsection{Fano's Inequality for Approximate Recovery}

As is common in studies of algorithm-independent lower bounds in learning problems, we make use of bounds based on Fano's inequality \cite[Sec.~2.10]{Cov01}.  We first briefly outline the most relevant results for the exact recovery problem.

Recall the definitions of $\pe$ and $\pe(\qmax)$ in \eqref{eq:pe}--\eqref{eq:pe_approx} with respect to a given graph class $\Gc$. It is known that for any subset $\Tc \subseteq \Gc$, and any covering set $C_{\Tc}(\epsilon)$ such that any graph $G \in \Tc$ has an ``$\epsilon$-close'' graph $G' \in C_{\Tc}(\epsilon)$ satisfying $D(P_G \| P_{G'}) \le \epsilon$, we have \cite{Sha14}
\begin{equation}
    \pe \ge 1 - \frac{\log |C_{\Tc}(\epsilon)| + n\epsilon + \log 2}{\log |\Tc|}. \label{eq:Fano_exact1}
\end{equation}
In particular, if $C_{\Tc}(\epsilon)$ is a singleton, solving for $n$ gives the necessary condition
\begin{equation}
    n \ge \frac{\log |\Tc|}{\epsilon}\Big( 1-\delta - \frac{\log 2}{\log|\Tc|}  \Big) \label{eq:Fano_exact2}
\end{equation}
in order to have $\pe \le \delta$.

For approximate recovery, we consider ensembles (i.e., choices of $\Tc$) for which the decoder's outputs may lie in some set $\Tc'$ without loss of optimality; in most cases we will have $\Tc = \Tc'$, but in general, $\Tc'$ need not even be a subset of the graph class $\Gc$.  We use the following generalization of \eqref{eq:Fano_exact2}.

\begin{lem} \label{lem:Fano_approx}
    Suppose that the decoder minimizing the average error probability with respect to a distortion level $\qmax$, averaged over a graph uniformly drawn from a set $\Tc \subseteq \Gc$, always outputs a graph in some set $\Tc'$.  Moreover, suppose that there exists a graph $G'$ such that $D(P_G \| P_{G'}) \le \epsilon$ for all $G \in \Tc$, and that there are at most $A(\qmax)$ graphs in $\Tc'$ within an edit distance $\qmax$ of any given graph $G \in \Tc$.  Then it is necessary that
    \begin{equation}
        n \ge \frac{\log |\Tc| - \log A(\qmax)}{\epsilon}\Big( 1-\delta - \frac{\log 2}{\log|\Tc|}  \Big) \label{eq:Fano_approx}
    \end{equation}
    in order to have $\pe(\qmax) \le \delta$.
\end{lem}
\begin{proof}
    See Appendix \ref{sec:PF_FANO}.
\end{proof}

\subsection{Properties of Ferromagnetic Ising Models}

We will use a number of useful results on ferromagnetic Ising models, each of which is either self-evident or can be found in \cite{San12} or \cite{Sha14}.  We start with some basic properties.

\begin{lem}
    For any graphs $G$ and $G'$ with edge sets $E$ and $E'$ respectively, we have the following:
    
    (i) For any pair $(i,j)$, we have \cite{San12}
    \begin{equation}
        \EE_G[X_i X_j] = 2\PP_G[X_i = X_j] - 1. \label{eq:L_EtoP}
    \end{equation}
    
    (ii) The divergence between the corresponding distributions satisfies \emph{\cite[Eq.~(4)]{Sha14}}
    \begin{align}
         D(P_G \| P_{G'}) 
             &\le \sum_{(i,j) \in E \backslash E'} \lambda\big( \EE_G[X_i X_j] - \EE_{G'}[X_i X_j] \big) \nonumber \\
                  & \quad + \sum_{(i,j) \in E' \backslash E} \lambda\big( \EE_{G'}[X_i X_j] - \EE_G[X_i X_j] \big) \label{eq:L_div} \\
             &\le \sum_{(i,j) \in E \backslash E'} \lambda\big( 1 - \EE_{G'}[X_i X_j] \big) \nonumber \\
                 &\quad + \sum_{(i,j) \in E' \backslash E} \lambda\big( 1 - \EE_G[X_i X_j] \big). \label{eq:L_div2}
    \end{align}
    
    (iii) If $E' \subset E$, then we have for any pair $(i,j)$ that \emph{\cite[Eq.~(13)]{Sha14}}
    \begin{equation}
        \EE_{G}[X_i X_j] \ge \EE_{G'}[X_i X_j]. \label{eq:L_subgraph}
    \end{equation}
    
    (iv) Let $(V_1,\dotsc,V_K)$ be a partition of $V$ into $K$ disjoint non-empty subsets. If $G$ and $G'$ are such that there are no edges between nodes in $V_i$ and $V_j$ when $i \ne j$, then 
    \begin{equation}
        D(P_G \| P_{G'}) = \sum_{i=1}^K D(P_{G_i} \| P_{G'_i}), \label{eq:L_disjoint}
    \end{equation}
    where $G_i = (V,E_i)$, with $E_i$ containing the edges in $E$ between nodes in $V_i$ (and analogously for $G'_i$).    
\end{lem}

The remaining properties concern the probabilities, expectations and divergences associated with more specific graphs.    
    
\begin{lem}
    (i) If $G'$ is obtained from $G$ by removing a single edge $(i,j)$, then \emph{\cite[Eq.~(19)]{Sha14}}
    \begin{equation}
        \frac{\PP_G[X_i = X_j]}{1-\PP_G[X_i = X_j]} = e^{2\lambda} \frac{\PP_{G'}[X_i = X_j]}{1-\PP_{G'}[X_i = X_j]} \label{eq:L_single_edge_p}
    \end{equation}
    and \emph{\cite[Lemma 4]{Sha14}}
    \begin{equation}
        D(P_G \| P_{G'}) \le \lambda\tanh\lambda. \label{eq:L_single_edge_div}
    \end{equation}
    
    (ii) Let $G$ contain a clique on $m \ge 2$ nodes and no other edges, and let $G'$ be obtained from $G$ by removing a single edge $(i,j)$.  Then, defining $\mbar := m - 1$, we have \emph{\cite[Eq.~(31)]{San12}}
    \begin{multline}
        \frac{ \PP_{G'}[X_i = X_j] }{1-\PP_{G'}[X_i = X_j]} \\ = \frac{ \sum_{j=0}^{\mbar} {\mbar \choose j} \exp\big( \frac{\lambda}{2}(2j - \mbar)^2 \big) \exp\big( 2\lambda(2j - \mbar) \big) }{ \sum_{j=0}^{\mbar} {\mbar \choose j} \exp\big( \frac{\lambda}{2}(2j - \mbar)^2 \big) }. \label{eq:L_clique_p}
    \end{multline}
    Moreover, we have \emph{\cite[Lemma 1]{San12}}
    \begin{equation}
        \PP_{G'}[X_i = X_j] \ge 1 - \frac{\mbar}{\mbar + e^{\mbar\lambda/2}} \label{eq:L_clique_p2}
    \end{equation}
    and
    \begin{equation}
        \EE_{G'}[X_i X_j] \ge 1 - \frac{2\mbar e^{\lambda}}{e^{\mbar\lambda} + \mbar e^{\lambda}}.  \label{eq:L_clique_e}
    \end{equation}
    
    (iii) Suppose that for some edge $(i,j) \in E \Delta E'$, there exist at least $m$ node-disjoint paths of length $\ell$ between $i$ and $j$ in $G$.  Then \emph{\cite[Lemma 3]{Sha14}}
    \begin{equation}
        \EE_G[X_iX_j] \ge 1 - \frac{ 2 }{ 1 + \Big( \frac{1 + (\tanh\lambda)^{\ell}}{1 - (\tanh\lambda)^{\ell}} \Big)^m }. \label{eq:L_path_e}
    \end{equation}
    If the same is true in both $G$ and $G'$ for all $(i,j) \in E \Delta E'$, then \emph{\cite[Cor.~3]{Sha14}}
    \begin{equation}
        D(P_G \| P_{G'}) \le \frac{ 2\lambda|E \Delta E'| }{ 1 + \Big( \frac{1 + (\tanh\lambda)^{\ell}}{1 - (\tanh\lambda)^{\ell}} \Big)^m }. \label{eq:L_path_div}
    \end{equation}
    
    (iv) More generally, if there exist at least $m_l$ node-disjoint paths of length $\ell_l$ between $(i,j)$ for $l=1,\dotsc,L$, where the values of $\ell_l$ are all distinct, then
    \begin{equation}
        \EE_G[X_iX_j] \ge 1 - \frac{ 2 }{ 1 + \prod_{l=1}^L \Big( \frac{1 + (\tanh\lambda)^{\ell_l}}{1 - (\tanh\lambda)^{\ell_l}} \Big)^{m_l} }. \label{eq:L_paths}
    \end{equation}
\end{lem}

\section{Graph Ensembles and Lower Bounds on their Sample Complexities} \label{sec:ENSEMBLES}

In this section, we provide necessary conditions for the approximate recovery of a number of ensembles, making use of the tools from the previous section.  In particular, we seek choices of $\Tc$, $\Tc'$ and $A(\qmax)$ for substitution into Fano's inequality in Lemma \ref{lem:Fano_approx}.  In Section \ref{sec:APPLICATIONS}, we use these to establish our main theorems.

\subsection{Ensemble 1: Many Isolated Edges}

\noindent This ensemble contains numerous isolated edges, such that if $\lambda$ is small then it is difficult to determine precisely which ones are present. It is constructed as follows with some integer parameter $\alpha \le p/4$:

\medskip
\noindent\fbox{
    \parbox{0.95\columnwidth}{
        \textbf{Ensemble1($\alpha$) [Isolated edges ensemble]:}
        \begin{itemize}
            \item Each graph in $\Tc$ is obtained by forming exactly $\alpha$ node-disjoint edges that may otherwise be arbitrary.
        \end{itemize}
    }
} \medskip

For this ensemble, we have the following properties:
\begin{itemize}
    \item The number of graphs is $|\Tc| = \prod_{i=0}^{\alpha} {p - 2i \choose 2} \ge {\lfloor p/2 \rfloor \choose 2}^{\alpha}$, since $p-2\alpha \ge \frac{p}{2}$ by the assumption $\alpha \le p/4$.
    \item The maximum degree of each graph is one.
    \item For this ensemble, it suffices to trivially let $\Tc'$ contain all graphs.
    \item The number of graphs within an edit distance $\qmax$ of any single graph is upper bounded as $A(\qmax) \le \sum_{q=0}^{\qmax} \sum_{q'=0}^{\qmax - q} {\alpha \choose q} {p \choose 2}^{q'} \le (1 + \qmax)^2{\alpha \choose \lfloor\alpha/2 \rfloor} {p \choose 2}^{\qmax}$.  Here the term ${\alpha \choose q}$ corresponds to choosing $q$ edges to remove, and the term ${p \choose 2}^{q'}$ upper bounds the number of ways to add $q' \le \qmax - q$ new edges.  We have also used the fact that ${\alpha \choose q}$ is maximized at $q = \lfloor \alpha/2 \rfloor$.
    \item From \eqref{eq:L_single_edge_div}, the KL divergence from a single-edge graph to the empty graph is upper bounded by $\lambda \tanh\lambda$.  Using this fact along with \eqref{eq:L_disjoint}, any graph in $\Tc$ has a KL divergence to the empty graph of at most $\epsilon = \alpha\lambda\tanh\lambda$.
\end{itemize}
Combining these with \eqref{eq:Fano_approx} gives the necessary condition
\begin{multline}
    n \ge \frac{\alpha \log {\lfloor p/2 \rfloor \choose 2} - \log \Big( (1 + \qmax)^2 {\alpha \choose \lfloor \alpha/2 \rfloor} {p \choose 2}^{\qmax}\Big) }{ \alpha\lambda\tanh\lambda } \\ \times\Big( 1-\delta - \frac{\log 2}{|\Tc|} \Big) \label{eq:ens1_init}
\end{multline}
in order to have $\pe(\qmax) \le \delta$.

Simplifying both $\log { \lfloor p/2 \rfloor \choose 2}$ and $\log {p \choose 2}$ to $(2\log p)(1+o(1))$, and writing $\log{\alpha \choose \lfloor \alpha/2 \rfloor} \le \alpha \log 2 = o(\alpha \log p)$ as well as $\log(1+\qmax)^2 \le 2\log(1+\alpha) = o(\alpha \log p)$, we can simplify \eqref{eq:ens1_init} to
\begin{align}
    n &\ge \frac{2\alpha \log {p} - 2\qmax \log {p} }{ \alpha\lambda\tanh\lambda }\Big( 1-\delta - o(1) \Big), \label{eq:ens1_init2}
\end{align}
provided that $\alpha \to \infty$ and $\qmax \le (1 - \Omega(1))\alpha$. Letting $\qmax = \lfloor\theta_1 \alpha\rfloor$ for some $\theta_1 \in (0,1)$, this becomes
\begin{equation}
    n \ge \frac{2(1 - \theta_1)\log {p} }{ \lambda\tanh\lambda }\Big( 1-\delta - o(1) \Big). \label{eq:ens1}
\end{equation}

\subsection{Ensemble 2: Many Isolated Groups of Nodes}

\noindent As an alternative to Ensemble 1, this ensemble allows for significantly more edges, in particular permitting $k=\omega(p)$.  It is constructed as follows with integer parameters $m$ and $\alpha$:

\medskip
\noindent\fbox{
    \parbox{0.95\columnwidth}{
       \textbf{Ensemble2($m$,$\alpha$) [Isolated cliques ensemble]:}
        \begin{itemize}
            \item Form $\alpha$ fixed groups of nodes, each containing $m$ nodes.
            \item Each graph in $\Tc$ is formed by forming arbitrarily many edges within each group, but no edges between the groups.
        \end{itemize}
    }
} \medskip

For this ensemble, we have the following:
\begin{itemize}
    \item The number of nodes forming these groups is $m\alpha$.
    \item The total number of possible edges is $\alpha{m \choose 2}$, and hence the total number of graphs is $|\Tc| = 2^{\alpha{m \choose 2}}$.
    \item The maximal degree of each graph is at most $m-1$.
    \item The decoder can output an element of $\Tc$ without loss of optimality, since any inter-group edges declared to be present are guaranteed to be wrong.  Thus, we may set $\Tc' = \Tc$.
    \item The number of graphs within an edit distance $\qmax$ of any single graph is $A(\qmax) = \sum_{q=0}^{\qmax}{ \alpha{m \choose 2} \choose q } \le 1 + \qmax{ \alpha{m \choose 2} \choose \qmax  }$, assuming $\qmax \le \frac{1}{2} \alpha{m \choose 2}$.
    \item In Lemma \ref{lem:groups_ens} below, we show that the KL divergence of the graph associated with one group to the corresponding empty graph is upper bounded by ${m \choose 2} \lambda \frac{e^{2\lambda}\cosh(2\lambda m) - 1}{e^{2\lambda}\cosh(2\lambda m) + 1}$.  Hence, the KL divergence of any $G \in \Tc$ to the empty graph is upper bounded by $\epsilon = \alpha{m \choose 2} \lambda  \frac{e^{2\lambda}\cosh(2\lambda m) - 1}{e^{2\lambda}\cosh(2\lambda m) + 1}$ due to \eqref{eq:L_disjoint}.  
\end{itemize}
Substituting these into \eqref{eq:Fano_approx}, setting $\qmax = \lfloor\theta_2\alpha{m \choose 2}\rfloor$ for some $\theta_2 \in \big(0,\frac{1}{2}\big)$, and applying some simplifications, we obtain the following necessary condition for $\pe(\qmax) \le \delta$:
\begin{equation}
    n \ge \frac{ \log 2 - H_2(\theta_2) }{ \lambda \frac{e^{2\lambda}\cosh(2\lambda m) - 1}{e^{2\lambda}\cosh(2\lambda m) + 1} } \big(1 - \delta - o(1)\big), \label{eq:ens2}
\end{equation}
whenever $\alpha{m \choose 2} \to \infty$.  Note that the binary entropy function arises from the identity ${N \choose \lfloor\theta N\rfloor} = e^{nH_2(\theta)(1+o(1))}$ as $N\to\infty$.

It remains to prove the claim on the KL divergence, formalized as follows.

\begin{lem} \label{lem:groups_ens}
    Let $G$ denote an arbitrary graph with edges connected to at most $m \ge 2$ nodes, and let $G'$ be the empty graph.  Then, it holds that
    \begin{equation}
        D(P_G \| P_{G'}) \le  {m \choose 2} \lambda \frac{e^{2\lambda}\cosh(2\lambda m) - 1}{e^{2\lambda}\cosh(2\lambda m) + 1}. \label{eq:groups_div}
    \end{equation}
\end{lem}
\begin{proof}
    We prove the claim for the case that $G$ contains a single $m$-clique; the general case then follows in a similar fashion using \eqref{eq:L_subgraph}.  
    
    Let $\Gbar$ be obtained from $G$ by removing a single edge, say indexed by $(i,j)$.  Defining $q(G) := \PP_G[X_i = X_j]$ and $\mbar := m-1$, we have from \eqref{eq:L_single_edge_p} that
    \begin{equation}
        \frac{q(G)}{1-q(G)} = e^{2\lambda} \frac{q(\Gbar)}{1-q(\Gbar)}, \label{eq:one_edge_diff}  
   \end{equation} 
    and from \eqref{eq:L_clique_p} that
    \begin{equation}
        \frac{q(\Gbar)}{1-q(\Gbar)} = \frac{ \sum_{j=0}^{\mbar} {\mbar \choose j} \exp\big( \frac{\lambda}{2}(2j - \mbar)^2 \big) \exp\big( 2\lambda(2j - \mbar) \big) }{ \sum_{j=0}^{\mbar} {\mbar \choose j} \exp\big( \frac{\lambda}{2}(2j - \mbar)^2 \big) }.
    \end{equation}
    Noting the symmetry of the summands with respect to $j$ and $\mbar - j$, we obtain the following when $\mbar$ is odd (the case that $\mbar$ is even is handled similarly, leading to the same conclusion):
    \begin{align}
        &\frac{q(\Gbar)}{1-q(\Gbar)} \nonumber \\
            &= \frac{ \sum_{j=0}^{\lfloor\mbar/2\rfloor} {\mbar \choose j} \exp\big( \frac{\lambda}{2}(2j - \mbar)^2 \big) \cdot 2\cosh\big( 2\lambda(2j - \mbar) \big) }{ 2\sum_{j=0}^{\lfloor\mbar/2\rfloor} {\mbar \choose j} \exp\big( \frac{\lambda}{2}(2j - \mbar)^2 \big) } \\
            &\le \max_{j=0,\dotsc,\lfloor\mbar/2\rfloor} \cosh\big( 2\lambda(2j - \mbar) \big) \\
            &= \cosh\big( 2\lambda\mbar \big) \\
            &\le \cosh\big( 2\lambda m \big). \label{eq:groups_ens4}
    \end{align}
    Substituting \eqref{eq:groups_ens4} into \eqref{eq:one_edge_diff}, solving for $q(G)$, and converting from probability to expectation via \eqref{eq:L_EtoP}, we obtain
    \begin{equation}
        \EE_G[X_iX_j] \le \frac{e^{2\lambda}\cosh(2\lambda m) - 1}{e^{2\lambda}\cosh(2\lambda m) + 1}.
    \end{equation}
    The proof is concluded by substituting into \eqref{eq:L_div} and noting that $\EE_{G'}[X_iX_j]=0$,  $|E \backslash E'| = {m \choose 2}$, and $|E' \backslash E| = 0$.
\end{proof}

\subsection{Ensemble 3: Large Inter-Connected Cliques} \label{sec:ENS3}

\noindent This ensemble involves cliques with numerous edges between them, making it difficult to determine precisely which inter-clique connections are present, particularly for large cliques and large values of $\lambda$.  It is constructed as follows with integer parameters $m$ and $\alpha$:

\medskip
\noindent\fbox{
    \parbox{0.95\columnwidth}{
        \textbf{Ensemble3($m$,$\alpha$) [Inter-connected cliques ensemble]:}
        \begin{itemize}
            \item Construct a fixed ``building block'' as follows: Take an arbitrary subset of the $p$ vertices of size $2m$, split the $2m$ vertices into two sets of size $m$ each, fully connect each of those sets, and then put $m$ extra edges between the two sets in a fixed but arbitrary one-to-one fashion.
            \item Form $\alpha$ disjoint copies of this building block to obtain a base graph $G_0$.
            \item Each graph in $\Tc$ is formed by taking $G_0$ and adding an arbitrary number of additional edges between each pair of partially-connected cliques.  Thus, $G_0$ itself contains the fewest edges within $\Tc$, and the union of $\alpha$ cliques of size $2m$ contains the most edges.
        \end{itemize}
    }
} \medskip

An illustration of one building block is given in Figure \ref{fig:ens3}.

For this ensemble, we have the following:
\begin{itemize}
    \item The number of nodes forming these groups is $2m\alpha$, and the number of edges in each graph is upper bounded by $\alpha {2m \choose 2} \le 2\alpha m^2$.
    \item The number of potential edges between two $m$-cliques is $m^2$, and $m$ of them are always there in each building block.  Hence, the number of ways of adding edges to one building block is $2^{m(m - 1)}$, and the total number of graphs is $2^{\alpha m(m - 1)}$.
    \item The maximal degree of each graph is at most $2m-1$.
    \item Similarly to Ensemble 2, the decoder can output an element of $\Tc$ without loss of optimality, so that $\Tc' = \Tc$.
    \item The number of graphs within an edit distance $\qmax$ of any single graph is $A(\qmax)=\sum_{q=0}^{\qmax}{ \alpha m(m - 1) \choose q } \le 1 + \qmax{ \alpha m(m - 1) \choose \qmax  }$, assuming $\qmax \le \frac{1}{2} \alpha m(m-1)$.
    \item In Lemma \ref{lem:large_clique_ens} below, we show that the KL divergence of the graph associated with one group to the $2m$-clique graph is upper bounded by $12\lambda m^4 e^{-\lambda(m-1)/2}$.  Thus, the KL divergence from any $G \in \Tc$ to the union of $\alpha$ $2m$-cliques is upper bounded by $\epsilon = 12\lambda \alpha m^4 e^{-\lambda(m-1)/2}$ due to \eqref{eq:L_disjoint}.
\end{itemize}

\begin{figure}
    \begin{centering}
        \hspace{3cm}\includegraphics[width=0.8\columnwidth]{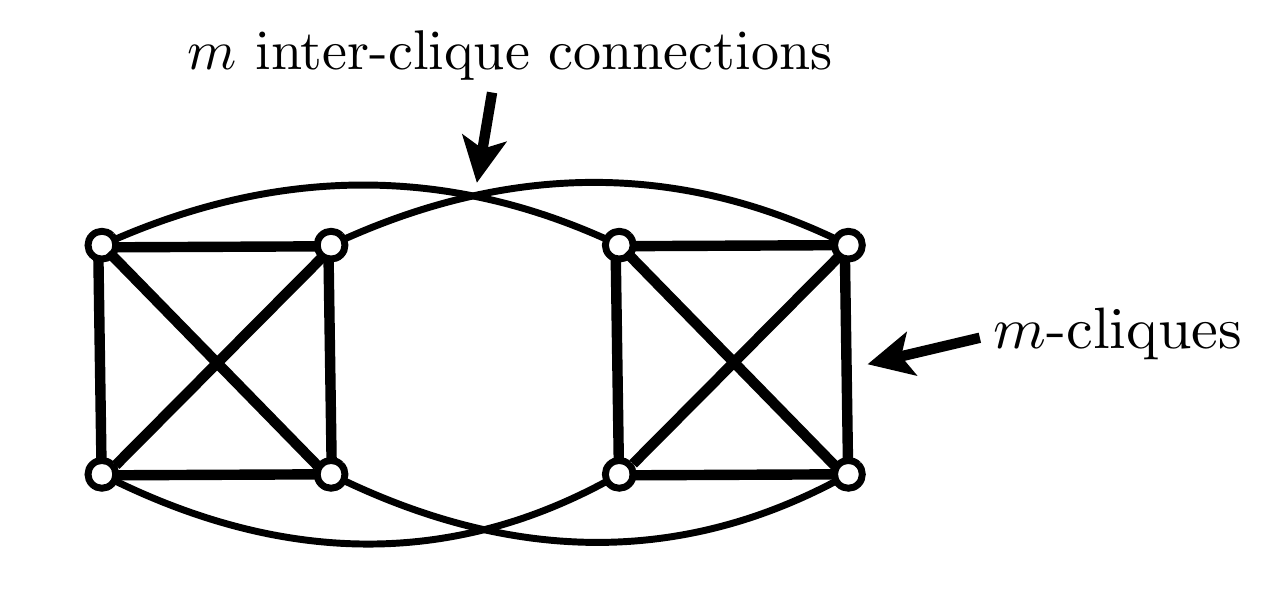}
        \par
    \end{centering}
    
    \caption{Building block for Ensemble 3 with $m = 4$. \label{fig:ens3}}
\end{figure}

Substituting these into \eqref{eq:Fano_approx}, setting $\qmax = \lfloor\theta_3\alpha m(m-1)\rfloor$ for some $\theta_3 \in \big(0,\frac{1}{2}\big)$, and simplifying, we obtain
\begin{equation}
    n \ge \frac{ e^{\lambda(m-1)/2} \big( \log 2  - H_2(\theta_3) \big) }{ 12\lambda m^2 }\big(1 - \delta - o(1)\big), \label{eq:ens3}
\end{equation}
whenever $\alpha m(m-1) \to\infty$.

It remains to prove the claim on the KL divergence, formalized as follows.

\begin{lem} \label{lem:large_clique_ens}
    Let $G$ denote the graph corresponding to a single group in Ensemble 3, and let $G'$ be the corresponding graph containing a $2m$-clique.  Then
    \begin{equation}
        D(P_G \| P_{G'}) \le 12\lambda m^4 e^{-\lambda(m-1)/2}. \label{eq:clique_div}
    \end{equation}
\end{lem}
\begin{proof}
    We focus on the case that $G$ is the building block obtained by forming two cliques of size $m$ and connecting $m$ edges between them; the case that further edges are present is handled similarly using \eqref{eq:L_subgraph}.

    From \eqref{eq:L_subgraph} and \eqref{eq:L_clique_p2}, we have for any $(i,j)$ within either of the two $m$-cliques that
    \begin{align}
        \PP_G[X_i = X_j] 
            &\ge 1 - \frac{\mbar}{\mbar + e^{\mbar\lambda/2}} \\
            &\ge 1 - \frac{m}{m + e^{\mbar\lambda/2}} ,
    \end{align}
    where $\mbar := m-1$.  By taking an arbitrary node from each clique and applying the union bound over the $2(m-1) \le 2m$ events corresponding to other nodes in the clique having the same value as that node, we find that the probability that each of the cliques have nodes that all take the same value satisfies the following:
    \begin{equation}
        \PP_G[\text{all nodes same within each clique}] \ge 1 - \frac{2m^2}{m + e^{\mbar\lambda/2}}. \label{eq:pr_same_cliques}
    \end{equation}
    Next, we consider the probabilities of the two cliques taking a common value vs.~two different values.  Letting $A_{\nu,\sigma}$ be the event that the $\nu$-th clique has values all equal to $\sigma \in \{+1,-1\}$, we have from \eqref{eq:Ising} that
    \begin{align}
        \PP_G[A_{1,+} \cap A_{2,+}] = \frac{1}{Z} \exp\bigg( \lambda\bigg( 2{m \choose 2} + m \bigg) \bigg) \\ 
        \PP_G[A_{1,+} \cap A_{2,-}] = \frac{1}{Z} \exp\bigg( \lambda\bigg( 2{m \choose 2} - m \bigg) \bigg).
    \end{align}
    Taking the ratio between the two gives
    \begin{equation}
        \frac{ \PP_G[A_{1,+} \cap A_{2,+}] }{ \PP_G[A_{1,+} \cap A_{2,-}] } = e^{2m\lambda}. \label{eq:clique_ratio}
    \end{equation}
    By the same argument, this is also the ratio between any analogous events with the same signs in the numerator and differing signs in the denominator.  The same argument also applies when we condition on each of the two cliques having common-valued nodes; in this case, the left-hand side of \eqref{eq:clique_ratio} simply amounts to $\frac{\psi}{1-\psi}$, where $\psi$ is the conditional probability that all of the $2m$ nodes making up the two cliques take the same value.  Equating $\frac{\psi}{1-\psi} = e^{2m\lambda}$ in accordance with \eqref{eq:clique_ratio} and solving for $\psi$, we obtain the following:
    \begin{multline}
        \PP_G[\text{all nodes same} \,|\, \text{all nodes same within each clique}] \\ = 1 - \frac{1}{1 + e^{2m\lambda}},
    \end{multline}
    where ``all nodes'' refers to the $2m$ nodes making up the two cliques.  Multiplying this with \eqref{eq:pr_same_cliques} gives
    \begin{equation}
        \PP_G[\text{all nodes same}] \ge 1 - \frac{2m^2}{m + e^{\mbar\lambda/2}} - \frac{1}{1 + e^{2m\lambda}}.
    \end{equation}
    Using this fact along with \eqref{eq:L_EtoP}, we have for all $(i,j)$, even in different cliques, that
    \begin{equation}
        \EE_G[X_i X_j] \ge 1 - \frac{4m^2}{m + e^{\mbar\lambda/2}}  - \frac{2}{1 + e^{2m\lambda}}.
    \end{equation}
    Finally, the number of edges that are in the complete graph $G'$ but not in $G$ is trivially upper bounded by ${2m \choose 2} \le 2m^2$, and thus substitution into \eqref{eq:L_div2} yields 
    \begin{equation}
        D(P_G \| P_{G'}) \le 2\lambda m^2 \bigg( \frac{4m^2}{m + e^{\lambda(m-1)/2}} + \frac{2}{1 + e^{2\lambda m}}  \bigg). 
    \end{equation}
    The proof is concluded by writing
    \begin{align}
        \frac{4m^2}{m + e^{\lambda(m-1)/2}} + \frac{2}{1 + e^{2\lambda m}} 
            &\le \frac{4m^2}{e^{\lambda(m-1)/2}} + \frac{2}{e^{2\lambda m}} \\
            &\le \frac{6m^2}{e^{\lambda(m-1)/2}}.
    \end{align}
\end{proof}

\begin{rem} \label{rem:extension}
    In this ensemble, there are $\alpha m^2$ edges known with certainty, and a possible further $\alpha m(m-1)$ that are unknown.  Thus, slightly more than half of the potential edges are known.  This limits the values of $\qmax$ that are meaningful when applying this ensemble, and is the reason for the constraints on $\qmax$  (e.g., $\qmax \le k/4$) in Theorems \ref{thm:necc_k1}--\ref{thm:necc_d2}.  However, one can generalize this ensemble by considering \emph{more than two} groups of $m$-cliques such that each pair has $m$ inter-clique connections.  With this extension, the fraction of potential edges that are known can be made arbitrarily close to zero, and similar results to those shown in Table \ref{tbl:summary} for $\Gc_k$ (respectively, $\Gc_{k,d}$) can be obtained even when $\qmax = \lfloor \theta \frac{k}{2} \rfloor$ (respectively, $\qmax = \lfloor \theta \frac{k}{2}\frac{d-2}{d} \rfloor$) for some $\theta \in (0,1)$.
\end{rem}

\subsection{Ensemble 4: Many Node-Disjoint Paths}

\noindent This ensemble is based on forming a large number of node-disjoint paths between pairs of nodes, making it difficult to determine whether or not \emph{direct} edges also exist between those nodes \cite{Sha14}.  It is constructed as follows with integer parameters $\eta_1$, $\eta_2$, $m$, $\ell$, $\alpha$:

\medskip
\noindent\fbox{
    \parbox{0.95\columnwidth}{
        \textbf{Ensemble4($\eta_1$,$\eta_2$,$m$,$\ell$,$\alpha$) [Disjoint paths ensemble]:}
        \begin{itemize}
            \item Take an arbitrary subset of the $p$ vertices of size $\eta_1$ and label them $1,2,\dotsc,\eta_1$.  For each consecutive pair of these nodes, including the wrapped-around pair $(\eta_1,1)$, form $\eta_2$ node-disjoint paths of length two between them, and also form $m$ node-disjoint paths of length $\ell$ between them.
            \item Form a base graph $G_0$ by taking $\alpha$ copies of this graph.
            \item Each graph in $\Tc$ is formed by taking $G_0$ and adding arbitrarily many edges among the $\eta_1$ ``center'' nodes of each building block.  Thus, $G_0$ itself has the fewest edges, whereas the graph with $\alpha{\eta_1 \choose 2}$ additional center edges contains the most edges.
        \end{itemize}
    }
} \medskip

An illustration of one building block is shown in Figure \ref{fig:ens4}.
\begin{figure}
    \begin{centering}
        \includegraphics[width=0.95\columnwidth]{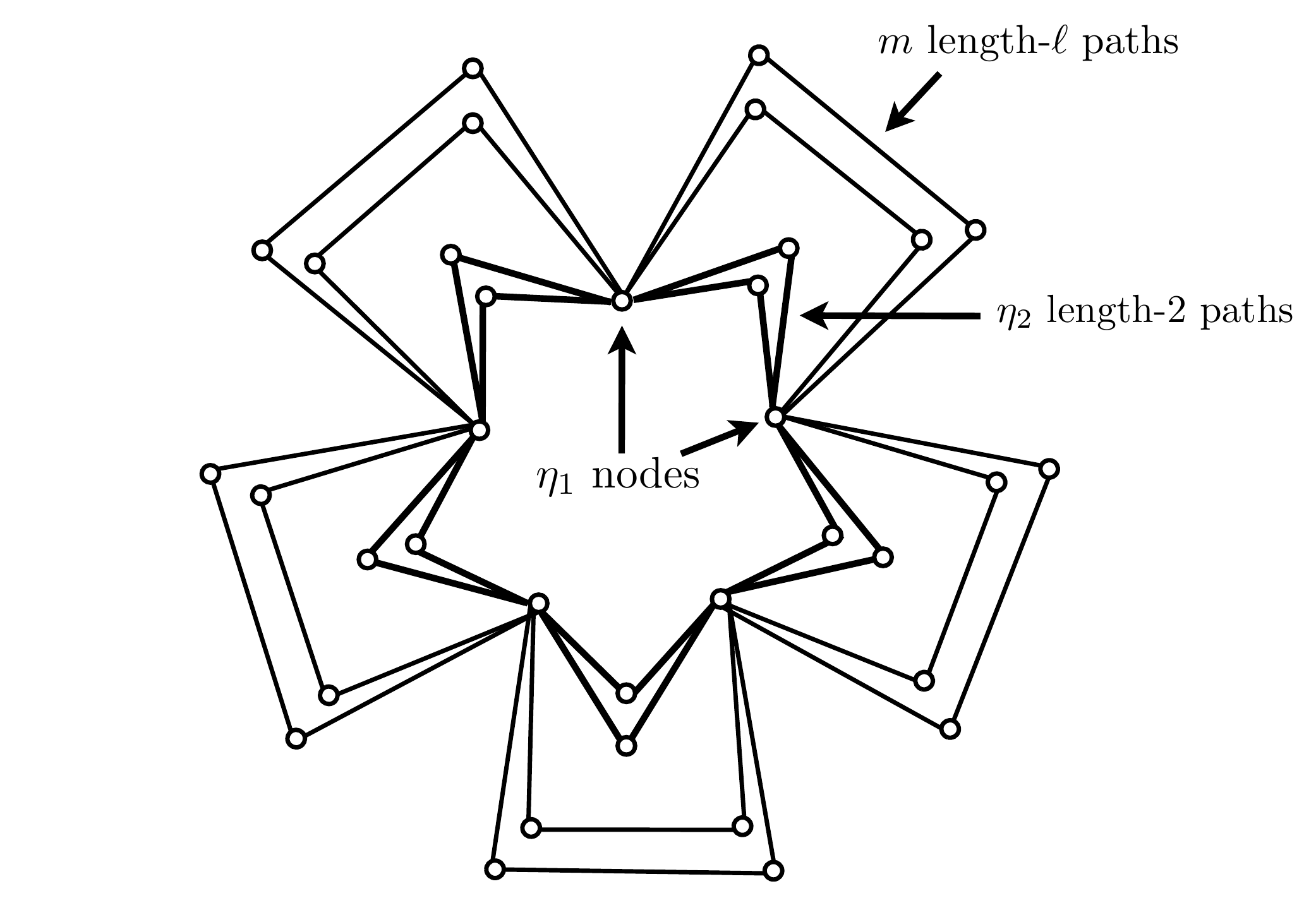}
        \par
    \end{centering}
    
    \caption{Building block for Ensemble 4 with $\eta_1 = 5$, $\eta_2$ = 2, $m = 2$, and $\ell = 3$. \label{fig:ens4}}
\end{figure}

\noindent For this ensemble, have the following:
\begin{itemize}
    \item The number of nodes within each building block is $\eta_1(1 + \eta_2 + m(\ell-1))$, and hence the total number of nodes is $\alpha\eta_1(1 + \eta_2 + m(\ell-1))$.
    \item Within each building block, there are up to $\eta_1 \choose 2$ edges in the center, as well as $2\eta_1\eta_2$ further edges forming paths of length two, and $m\eta_1\ell$ edges forming paths of length $\ell$.   Hence, the total number of edges is between $\alpha\eta_1(2\eta_2+m\ell)$ and $\alpha\eta_1((\eta_1-1)/2 + 2\eta_2+m\ell)$.
    \item The total number of graphs is $|\Tc| = 2^{\alpha{\eta_1 \choose 2}}$.
    \item The maximal degree is less than $\eta_1 + 2\eta_2 + 2m$. 
    \item Similarly to Ensembles 2 and 3, we may set $\Tc' = \Tc$.
    \item The number of graphs within an edit distance $\qmax$ of any given graph is $A(\qmax) = \sum_{q=0}^{\qmax} {{ \alpha{\eta_1 \choose 2} } \choose q } \le 1 + \qmax {{ \alpha{\eta_1 \choose 2} } \choose \qmax}$, assuming $\qmax \le \frac{1}{2}\alpha{\eta_1 \choose 2}$.
    \item Using Lemma \ref{lem:path_ens} below, along with \eqref{eq:L_disjoint}, the KL divergence from any graph in $\Tc$ to the corresponding graph with all centers connected is upper bounded by $\epsilon = \frac{ 2\lambda \alpha \eta_1 {\eta_1 \choose 2} }{ 1 + \big(\cosh(2\lambda)\big)^{\eta_2} \big( \frac{1+(\tanh\lambda)^{\ell}}{1-(\tanh\lambda)^{\ell}} \big)^m }$. 
\end{itemize}
Substituting these into \eqref{eq:Fano_approx} and setting $\qmax = \lfloor \theta_4 \alpha{\eta_1 \choose 2} \rfloor$ for some $\theta_4 \in \big(0,\frac{1}{2}\big)$ gives 
\begin{multline}
    n \ge \frac{ 1 + \big(\cosh(2\lambda)\big)^{\eta_2} \big( \frac{1+(\tanh\lambda)^{\ell}}{1-(\tanh\lambda)^{\ell}} \big)^m }{ 2\lambda \eta_1 } \big( \log 2 - H_2(\theta_4) \big) \\ \times\big(1 - \delta - o(1)\big) \label{eq:ens4}
\end{multline}
provided that $\alpha{\eta_1 \choose 2} \to \infty$.

It remains to prove the claim on the KL divergence, formalized as follows.

\begin{lem} \label{lem:path_ens}
    Let $G$ denote the graph corresponding to a single group in the construction in Ensemble 4, and let $G'$ be the corresponding building block with all of the center nodes connected.  Then
    \begin{equation}
        D(P_G \| P_{G'}) \le \frac{ 2\lambda \eta_1 {\eta_1 \choose 2} }{ 1 + \big(\cosh(2\lambda)\big)^{\eta_2} \big( \frac{1+(\tanh\lambda)^{\ell}}{1-(\tanh\lambda)^{\ell}} \big)^m }. \label{eq:path_div}
    \end{equation}
\end{lem}
\begin{proof}
    We focus on the case that $G$ is the building block described above; the case that further edges are present is handled similarly using \eqref{eq:L_subgraph}.

    We know from \eqref{eq:L_paths} that the joint distribution between any two consecutive nodes in the center satisfies
    \begin{equation}
        \EE_G[X_i X_j] \ge 1 - \frac{2}{1 + \big(\cosh(2\lambda)\big)^{\eta_2} \big( \frac{1+(\tanh\lambda)^{\ell}}{1-(\tanh\lambda)^{\ell}} \big)^m},
    \end{equation}
    since $\frac{1 + (\tanh\lambda)^2}{1 - (\tanh\lambda)^2} = \cosh(2\lambda)$.  Using \eqref{eq:L_EtoP}, this implies
    \begin{equation}
        \PP_G[X_i = X_j] \ge 1 - \frac{1}{1 + \big(\cosh(2\lambda)\big)^{\eta_2} \big( \frac{1+(\tanh\lambda)^{\ell}}{1-(\tanh\lambda)^{\ell}} \big)^m}.
    \end{equation}
    Thus, by applying the union bound over $(i,j)$ pairs of the form $(1,2),(2,3),\dotsc,(\eta_1-1,\eta_1),(\eta_1,1)$, the probability that all $\eta_1$ of the center nodes take the same value satisfies
    \begin{multline}
        \PP_G[\text{all center nodes same}] \\ \ge 1 - \frac{\eta_1}{1 + \big(\cosh(2\lambda)\big)^{\eta_2} \big( \frac{1+(\tanh\lambda)^{\ell}}{1-(\tanh\lambda)^{\ell}} \big)^m}.
    \end{multline}
    Again using \eqref{eq:L_EtoP}, this implies for any pair of center nodes $(i,j)$, including non-adjacent pairs, that
    \begin{equation}
        \EE_G[X_i X_j] \ge 1 - \frac{2\eta_1}{1 + \big(\cosh(2\lambda)\big)^{\eta_2} \big( \frac{1+(\tanh\lambda)^{\ell}}{1-(\tanh\lambda)^{\ell}} \big)^m}.
    \end{equation}
    Observing that the corresponding edge sets $E$ and $E'$ satisfy $|E'\backslash E| \le {\eta_1 \choose 2}$  and $|E\backslash E'| = 0$, \eqref{eq:path_div} follows from \eqref{eq:L_div2}.
\end{proof}

\section{Applications to Graph Families} \label{sec:APPLICATIONS}

Finally, we prove our main results by applying the ensembles from the previous section to the graph families introduced in Section \ref{sec:SETUP}.  All of the necessary conditions on $n$ stated in this section are those needed to obtain $\pe(\qmax) \le \delta$, where the graph class defining $\pe(\cdot)$ will be clear from the context.

\subsection{Proofs of Theorems \ref{thm:necc_k1}--\ref{thm:necc_k2}: Bounded Edges Ensemble}

For the class $\Gc_{k}$ of graphs with at most $k$ edges, we have the following:
\begin{itemize}
     \item If $k \le p/4$, then using Ensemble 1 with $\alpha = k$, we obtain from \eqref{eq:ens1} that
    \begin{equation}
        n \ge \frac{ 2(1-\theta_1)\log p }{ \lambda\tanh\lambda }\Big( 1-\delta - o(1) \Big) \label{eq:k_cond_1}
    \end{equation}
    provided that $\qmax \le \lfloor \theta_1 k \rfloor$ for some $\theta_1 \in (0,1)$.
    \item If $k = \lfloor cp^{1+\nu} \rfloor$ for some $c > 0$ and $\nu \in [0,1)$, then we use Ensemble 2 with $m = \lfloor 2cp^{\nu} \rfloor$ and $\alpha = \lfloor p/m \rfloor = \frac{1}{2c} p^{1-\nu} (1+o(1))$, chosen so that $m\alpha \le p$ nodes are used in the construction.  The number of possible edges is $\alpha{m \choose 2} \le \frac{1}{2}\alpha m^2 \le \frac{1}{2}pm \le c p^{1+\nu}$, as desired.  We obtain from \eqref{eq:ens2} that
    \begin{equation}
        n \ge \frac{ \log 2 - H_2(\theta_2) }{ \lambda \frac{e^{2\lambda}\cosh(2\lambda cp^{\nu}) - 1}{e^{2\lambda}\cosh(2\lambda cp^{\nu}) + 1} } \big(1 - \delta - o(1)\big) \label{eq:k_cond_1a}
    \end{equation}
    provided that $\qmax \le \lfloor \theta_2 \alpha {m \choose 2}\rfloor$ for some $\theta_2 \in \big(0,\frac{1}{2}\big)$.  Substituting the choices of $m$ and $\alpha$ into the latter expression, we find that $\qmax$ can be as large as $\theta_2 k (1+o(1))$.
    \item We use Ensemble 3 with $\alpha=1$ and $m = \lfloor \sqrt{k/2} \rfloor$, chosen so that the number of edges does not exceed $2\alpha m^2 \le k$. With these choices, we obtain from \eqref{eq:ens3}, along with the identity $\lfloor m \rfloor \ge m - 1$, that
    \begin{equation}
        n \ge \frac{ e^{\lambda(\sqrt{k/2}-2)/2} \big( \log 2  - H_2(\theta_3) \big) }{ 6\lambda k } \big(1 - \delta - o(1)\big), \label{eq:k_cond_2}
    \end{equation} 
     provided that $\qmax \le \lfloor\theta_3 \alpha m(m-1)\rfloor$ for some $\theta_3 \in \big(0,\frac{1}{2}\big)$.  Substituting the choices of $m$ and $\alpha$ into the latter expression, we find that $\qmax$ can be as large as $\theta_3 \frac{k}{2} (1+o(1))$, provided that $k \to \infty$.  Note that this construction uses $2m\alpha \le \sqrt{2k}$ nodes, which is asymptotically less than $p$ since $k = o(p^2)$.
\end{itemize}
We obtain Theorem \ref{thm:necc_k1} from \eqref{eq:k_cond_1} and \eqref{eq:k_cond_2}, and Theorem \ref{thm:necc_k2} from \eqref{eq:k_cond_1a} and \eqref{eq:k_cond_2}.  Specifically, we set $\qmax = \lfloor \theta k \rfloor$ for some $\theta \in \big(0,\frac{1}{4}\big)$, and by equating this with the above upper bounds on $\qmax$ we see that we may set $\theta_1 = \theta$, $\theta_2 = \theta(1+o(1))$ and $\theta_3 = 2\theta(1+o(1))$.

\subsection{Proofs of Theorems \ref{thm:necc_d1}--\ref{thm:necc_d2}: Bounded Degree Ensemble}

For the class $\Gc_{k,d}$ of graphs such that every node has degree at most $d$, and the total number of edges does not exceed $k$, we have the following:
\begin{itemize}
     \item If $k \le p/4$, then using Ensemble 1 with $\alpha = k$, we obtain from \eqref{eq:ens1} that
    \begin{equation}
        n \ge \frac{ 2(1-\theta_1)\log p }{ \lambda\tanh\lambda }\Big( 1-\delta - o(1) \Big), \label{eq:d_cond_1}
    \end{equation}
    provided that $\qmax \le \lfloor \theta_1 k \rfloor$ for some $\theta_1 \in (0,1)$.
    \item In the case that $k = \Omega(p)$, we use Ensemble 2 with the following parameters:
    \begin{enumerate}
        \item $m = d' \le d$, chosen so that the maximal degree $m-1$ does not exceed $d$;
        \item $\alpha = \lfloor k/{d' \choose 2} \rfloor$, chosen so that the number of edges $\alpha{m \choose 2}$ does not exceed $k$.
    \end{enumerate} 
    With these choices, we obtain from \eqref{eq:ens2} that
    \begin{equation}
        n \ge \frac{ \log 2 - H_2(\theta_2) }{ \lambda \frac{e^{2\lambda}\cosh(2\lambda d') - 1}{e^{2\lambda}\cosh(2\lambda d') + 1} } \big(1 - \delta - o(1)\big), \label{eq:d_cond_1a}
    \end{equation}
    whenever $\qmax \le \lfloor \theta_2 \alpha {d' \choose 2} \rfloor$ for some $\theta_2 \in \big(0,\frac{1}{2}\big)$. Substituting the choice of $\alpha$, we find that $\qmax$ can be a large as $\theta_2 k (1+o(1))$.  Note also that the number of nodes used is upper bounded as $\alpha m \le \frac{k}{{d' \choose 2}} d' = \frac{2k}{d'-1}$, which is upper bounded by $p$ provided that $k \le \frac{1}{2}p(d'-1)$.
    \item We use Ensemble 3 with the following parameters:
    \begin{enumerate}
        \item $m = \lceil d/2 \rceil$, chosen so that each block has nodes with degree not exceeding $2m-1 \le d$; 
        \item $\alpha = \big\lfloor \frac{k}{{2m \choose 2}} \big\rfloor$, chosen to ensure that the number of edges does not exceed $\alpha{2m \choose 2} \le k$.
    \end{enumerate}
    With these choices, we obtain from \eqref{eq:ens3} that
    \begin{equation}
        n \ge \frac{ e^{\lambda (d-2)/4} \big( \log 2 - H_2(\theta_3) \big) }{ 3\lambda d^2 } \big(1 - \delta - o(1)\big), \label{eq:d_cond_2}
    \end{equation}
    when $\qmax \le \lfloor\theta_3 \alpha m(m-1) \rfloor$ for some $\theta_3 \in \big(0,\frac{1}{2}\big)$. Substituting the choice of $\alpha$ to obtain $\alpha m(m-1) = k\frac{m-1}{2m-1} (1+o(1))$, and then writing $d/2  \le m \le (d+1)/2$, we find that the latter condition holds provided that $\qmax \le \frac{d/2 - 1}{d} \theta_3 k (1+o(1))$.  The number of nodes used is $2m\alpha \le \frac{2mk}{m(2m-1)} = \frac{2k}{2m-1} \le \frac{2k}{d-1}$, which is upper bounded by $p$ provided that $k \le \frac{1}{2}p(d-1)$.
\end{itemize}

We obtain Theorem \ref{thm:necc_d1} from \eqref{eq:d_cond_1} and \eqref{eq:d_cond_2}, and Theorem \ref{thm:necc_d2} from \eqref{eq:d_cond_1a} and \eqref{eq:d_cond_2}.  Similarly to the previous subsection, we set $\theta_1 = \theta$, $\theta_2 = \theta(1+o(1))$, and $\theta_3 = \frac{d}{d - 2}\cdot 2\theta(1+o(1))$.

\subsection{Proofs of Theorem \ref{thm:necc_s1}: Sparse Separator Ensemble}

For the class $\Gc_{k,d,\eta,\gamma}$ (\emph{cf.}~Section \ref{sec:SPARSE_ENS}), we have the following:
\begin{itemize}
    \item If $k \le p/4$, then again using Ensemble 1 with $\alpha = k$, we obtain from \eqref{eq:ens1_init2} that
        \begin{equation}
            n \ge \frac{ 2(k - \qmax) \log {p} }{ k\lambda\tanh\lambda }\big( 1-\delta - o(1) \big). \label{eq:sep_cond_1}
        \end{equation}
    \item We use Ensemble 4 with the following parameters:
    \begin{enumerate}
        \item $\eta_1 = \lfloor c \eta \rfloor$ and $\eta_2 = \lfloor (1- c)\eta \rfloor$ for some $c \in \big(\frac{1}{\eta},1\big]$, thus ensuring that $\eta_1 \ge 1$; 
        \item $\ell = \gamma+1$, chosen to ensure that the $(\eta,\gamma)$-separation condition is satisfied; 
        \item $m \le d/2 - \eta$, chosen so that the maximal degree is upper bounded by $\eta_1 + 2\eta_2 + 2m \le 2\eta + 2m \le d$;
        \item $\alpha = \lfloor \frac{k}{ c \eta( c\eta/2 + 2(1-c)\eta + m(\gamma + 1) } \rfloor$, chosen to ensure the total number of edges $\alpha\eta_1((\eta_1-1)/2 + 2\eta_2+m\ell)$ does not exceed $k$.
    \end{enumerate} 
    With these choices, we obtain from \eqref{eq:ens4} that
    \begin{multline}
        n \ge \frac{ 1 + \big(\cosh(2\lambda)\big)^{(1-c)\eta - 1} \big( \frac{1+(\tanh\lambda)^{\gamma+1}}{1-(\tanh\lambda)^{\gamma+1}} \big)^{m} }{ 2\lambda c \eta } \\ \times \big( \log 2 - H_2(\theta_4) \big)\big( 1-\delta - o(1) \big) \label{eq:sep_cond_2}
    \end{multline}
    provided that $\qmax \le \lfloor \theta_4 \alpha (c\eta - 1)^2/2 \rfloor$ for some $\theta_4\in\big(0,\frac{1}{2}\big)$.  Here we have used $\zeta - 1 \le \lfloor \zeta \rfloor \le \zeta$ and ${c\eta \choose 2} \ge (c\eta-1)^2/2$.  Note that the graph in this ensemble with the most edges has at least as many edges as nodes, since each node is connected to at least two edges.  Thus, since we have assumed $k \le p/4$ and we have already chosen the parameters to ensure there are at most $k$ edges, we have also ensured that less than $p$ nodes are used. Substituting the above choice of $\alpha$ into the upper bound on $\qmax$, we find that $\qmax$ can be as large as
    \begin{multline}
        \Big\lfloor \theta_4 \frac{(c\eta - 1)^2 k}{ 2c\eta( c\eta/2 + 2(1-c)\eta + m(\gamma + 1) } \Big\rfloor \\
                    \ge \Big\lfloor \theta_4 \frac{(c\eta - 1)^2 k}{ 2c\eta( 2\eta + m(\gamma + 1)) }\Big\rfloor, \label{eq:sep_qmax_bound}
    \end{multline}
    since $c\eta/2 + 2(1-c)\eta \le 2\eta$ for $c \in [0,1]$.  
\end{itemize}

We obtain Theorem \ref{thm:necc_s1} by combining \eqref{eq:sep_cond_1}, \eqref{eq:sep_cond_2} and \eqref{eq:sep_qmax_bound}, and renaming $\theta_4$ as $\theta$.

\section{Numerical Results} \label{sec:NUMERICAL}

In this section, we simulate the graph learning problem for some of the ensembles presented in Section \ref{sec:ENSEMBLES}, as well as the analogous ensembles used for exact recovery in \cite{San12,Sha14}.  Before proceeding, we discuss the optimal decoding techniques for the two recovery criteria.

Suppose that the graph $G$ is uniformly drawn from some class $\Gc$.  In the case of exact recovery, the optimal decoder is the maximum-likelihood (ML) rule
\begin{equation}
    \hat{G} = \argmax_{G \in \Gc} \PP_{G}[\Xv], \label{eq:ML}
\end{equation}
where $\PP_{G}[\Xv]$ is the probability of observing the samples $\Xv \in \{0,1\}^{n \times p}$ when the true graph is $G$.  In contrast, the optimal rule for approximate recovery is given by
\begin{equation}
    \hat{G} = \argmax_{G \in \Gc} \sum_{G' \,:\, |E\Delta E'| \le \qmax}  \PP_{G'}[\Xv], \label{eq:ML_partial}
\end{equation}
where $E$ and $E'$ are the edge sets of $G$ and $G'$ respectively.  Both \eqref{eq:ML} and \eqref{eq:ML_partial} are, in general, computationally intractable, requiring a search over the entire space $\Gc$.  However, in the examples below, we are able to apply \eqref{eq:ML} by using various tricks such as symmetry arguments.  While we need to consider relatively small graph sizes for Ensembles 3 and 4, these will still be adequate for generating results that support the theory.

Unfortunately, we found the implementation of \eqref{eq:ML_partial} much more difficult, and we therefore also use \eqref{eq:ML} for approximate recovery even though, in general, it is only optimal for exact recovery. Nevertheless, even with approximate recovery, we expect ML to provide a benchmark that that is unlikely to be beaten by any practical methods.

In all of the experiments, the error probabilities are obtained by evaluating the empirical average over $5000$ trials.

\subsection{A Variant of Ensemble 1 and a Counterpart from \cite{San12}}

It was shown in \cite{San12} that if one considers all graphs with a single edge, then it is difficult to distinguish each of these from the empty graph if $\lambda$ is small, thus making exact recovery difficult.  In Figure \ref{fig:numerical1}, we simulate the performance of this ensemble with $p=100$.  Since the partition function $Z$ (see \eqref{eq:Ising}) is the same for all graphs in this ensemble, the ML rule \eqref{eq:ML} simply amounts to declaring the single edge to be the pair $(i,j)$ among the ${p \choose 2}$ possibilities such that $X_i = X_j$ in the highest number of samples.

Our Ensemble 1 is analogous to the single-edge ensemble from \cite{San12}; however, in order to facilitate the computation, we consider a slight variant defined as follows:

\medskip
\noindent\fbox{
    \parbox{0.95\columnwidth}{
        \textbf{Ensemble1a($\alpha$) [Isolated edges ensemble]:}
        \begin{itemize}
            \item Group the $p$ vertices into $p/2$ fixed pairs in an arbitrary manner.
            \item Each graph in $\Tc$ is obtained by connecting exactly $\alpha$ of those $p/2$ pairs.
        \end{itemize}
    }
} \medskip

Note that Ensemble 1a can be interpreted as a genie-aided version of Ensemble 1, where the decoder is given information narrowing the $\prod_{i=0}^{\alpha} {p - 2i \choose 2}$ possible graphs down to a smaller set of size ${p/2 \choose \alpha}$.  For this reason, the performance under Ensemble 1a is an optimistic estimate of the performance under Ensemble 1, and moving to the latter should only narrow the gaps seen in our comparisons to \cite{San12}.

Figure \ref{fig:numerical1} plots the approximate recovery error probability for Ensemble 1a with $p=100$ and $\alpha = 12$, setting $\qmax = 3$ so that up to a quarter of the edges may be in error.  The maximum-likelihood rule \eqref{eq:ML} is simple to implement: Since all graphs have the same partition function, the most likely graph corresponds to choosing the $\alpha$ edges among the $p/2$ potential edges, such that the corresponding pairs of nodes agree in as many observations as possible.  This can  be implemented by simply counting the number agreements of the $p/2$ pairs and then sorting.

In accordance with our theory, the general behavior of the error probability as a function of $n$ is similar for Ensemble 1a (approximate recovery) and the ensemble from \cite{San12} (exact recovery).  Moving to approximate recovery does provide some benefit, but it appears to be only in the constant factors.  More specifically, across the range shown, the number of measurements required to achieve a given error probability in $[0.01, 0.5]$ differs for the two ensembles and recovery criteria only by a multiplicative factor in the range $[1,2.2]$. In both cases, the learning problem becomes increasingly difficult as $\lambda$ becomes smaller, since the edges are weaker and therefore more difficult to detect.

\begin{figure}
    \begin{centering}
        \includegraphics[width=0.95\columnwidth]{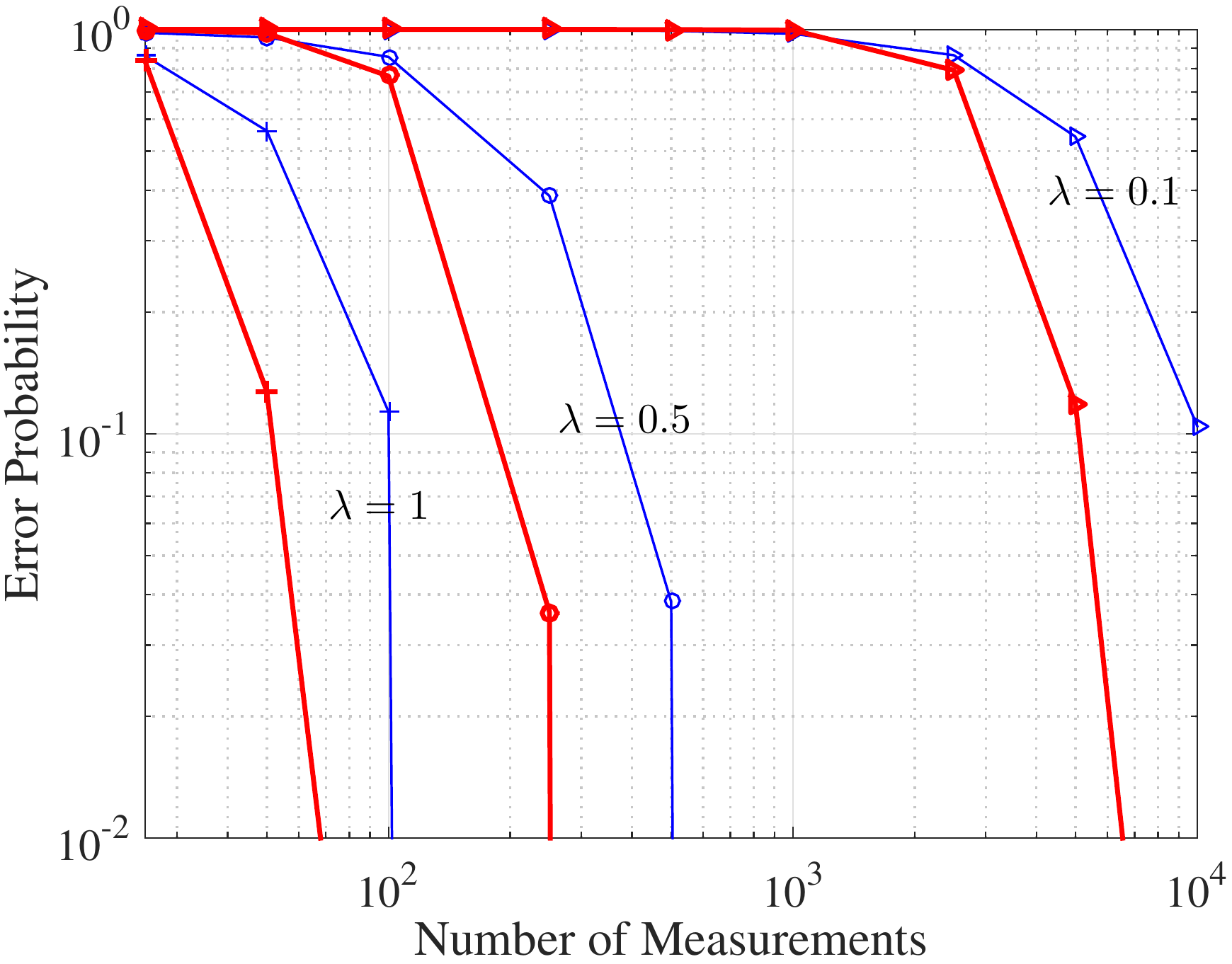}
        \par
    \end{centering}
    
    \caption{Empirical performance for Ensemble 1a (approximate recovery; red bold) and its counterpart from \cite{San12} (exact recovery; blue non-bold). \label{fig:numerical1}} \vspace*{-3ex}
\end{figure}

\subsection{Ensemble 3 and a Counterpart from \cite{San12}} 

A counterpart to Ensemble 3 from \cite{San12} considers the ${m' \choose 2}$ possible graphs on $m'$ nodes obtained by removing a single edge from the $m'$-clique.  Thus, every graph is difficult to distinguish from the $m'$-clique, particularly as $m'$ and $\lambda$ increase, and exact recovery is difficult.  In Figure \ref{fig:numerical3}, we plot the performance of this ensemble with $m' = 8$.  In this case, ML decoding amounts to choosing the pair $(i,j)$ such that $X_i \ne X_j$ in the highest number of samples.

For comparison, we consider Ensemble 3 with $m=4$ and $\alpha = 1$, chosen so that the maximal number of edges and degree match those of the ensemble from \cite{San12} with $m'=8$.  We set $\qmax = 3$, so that up to a quarter of the $12$ \emph{unknown} edges may be in error.  We perform ML decoding using a brute force search over the $2^{12}$ possible graphs.

Compared to the previous example, the gap between the curves for approximate recovery and exact recovery are more significant.  This is because although both our results and those of \cite{San12} prove that the sample complexity is exponential in $\lambda m$, the exponent in \cite{San12} is double that of ours.  Intuitively, this is because we work with cliques of half the size.  Despite this, the general behavior of our curves and those of \cite{San12} is similar, with the sample complexity rapidly growing large as $\lambda$ increases due to higher correlations among the $8$ nodes.

\begin{figure}
    \begin{centering}
        \includegraphics[width=0.95\columnwidth]{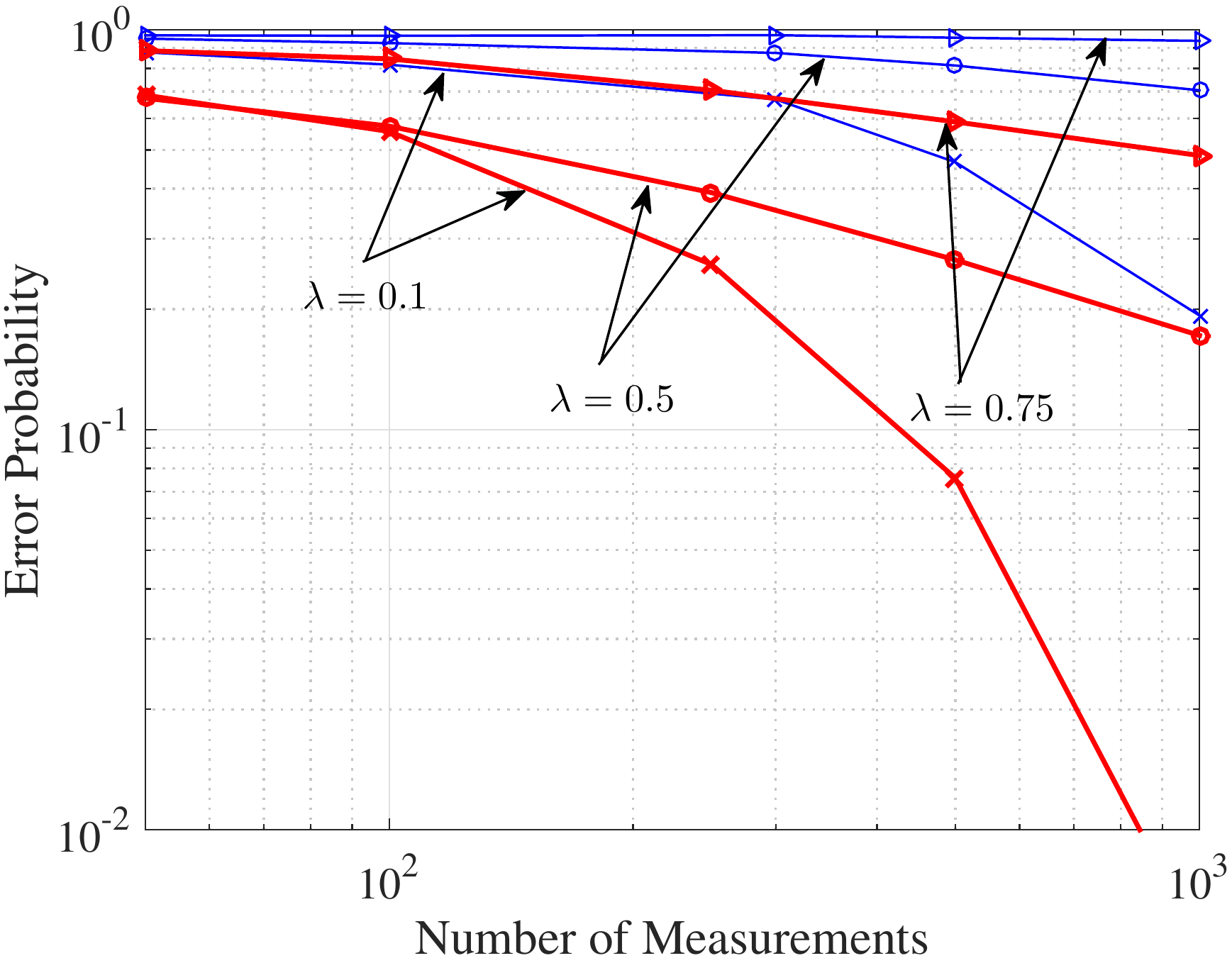}
        \par
    \end{centering}
    
    \caption{Empirical performance for Ensemble 3 (approximate recovery; red bold) and its counterpart from \cite{San12} (exact recovery; blue non-bold).  \label{fig:numerical3}} \vspace*{-3ex}
\end{figure}

\subsection{Ensemble 4 and a Counterpart from \cite{Sha14}}

A counterpart to Ensemble 4 from \cite{Sha14} first constructs $\alpha$ disjoint building blocks, each of which connects two nodes $(i,j)$, and then forms $\eta$ node-disjoint paths of length $2$ between them.  Each graph in the ensemble is then obtained by removing the direct edge from one of the $\alpha$ building blocks, while leaving the length-$2$ paths unchanged.  We consider this construction with $\alpha = 4$ and $\eta = 8$, thus leading to the use of $p = 40$ nodes and $k = 68$ edges, and a maximal degree $d = 9$.  Figure \ref{fig:numerical4} plots the performance of the ML decoder, which amounts to counting the number agreements between the $\alpha$ pairs of ``central'' nodes (one per building block), and declaring the edge to be absent in the one with the most disagreements.

For comparison, we consider Ensemble 4 with $\eta_1=4$, $\eta_2 = 3$, $m = 0$ and $\alpha = 2$; this construction uses $p = 32$ nodes and $k = 60$ edges, and has a maximal degree $d = 9$, thus being comparable to the above construction from \cite{Sha14}.  We set $\qmax = 3$, so that up to a quarter of the $12$ \emph{unknown} edges may be in error.  We perform ML decoding using a brute force search over the $2^{12}$ possible graphs, which simplifies to performing ML separately on the $2^6$ possible graphs corresponding to each of the two building blocks.

Once again, we observe the same general behavior between our ensemble and that of \cite{Sha14}.  While it may appear unusual that the exact recovery curves have a smaller error probability at low values of $n$, this occurs because even a random guess achieves an probability of exact recovery of $\frac{1}{4}$ for the ensemble in \cite{Sha14} with $\alpha = 4$. Despite this, we see that approximate recovery is easier for large $n$ as expected, and that in both cases the recovery problem rapidly becomes more difficult as $\lambda$ increases due to higher correlations among the nodes.

\begin{figure}
    \begin{centering}
        \includegraphics[width=0.95\columnwidth]{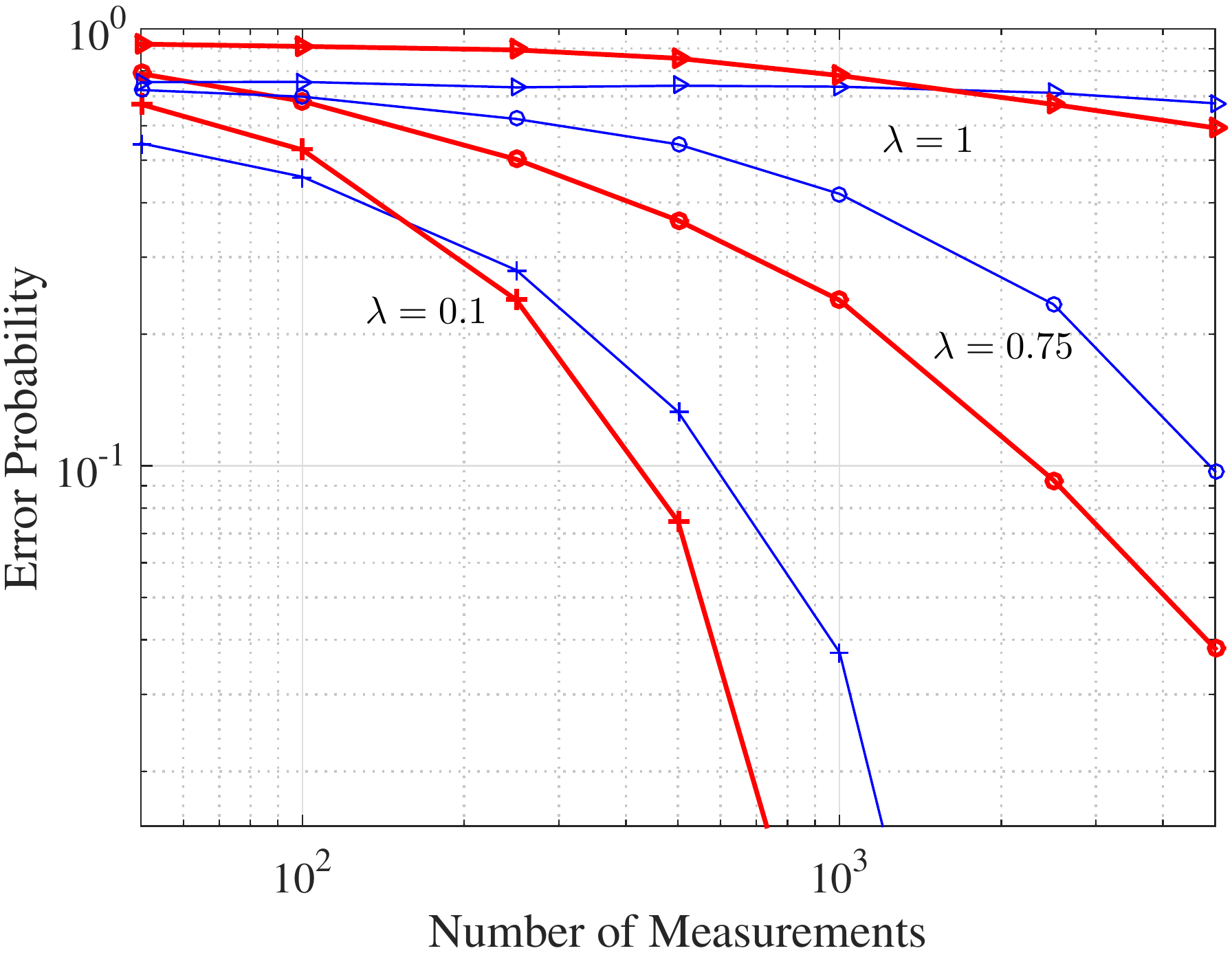}
        \par
    \end{centering}
    
    \caption{Empirical performance for Ensemble 4 (approximate recovery; red bold) and its counterpart from \cite{Sha14} (exact recovery; blue non-bold). \label{fig:numerical4}} \vspace*{-3ex}
\end{figure}

\section{Conclusion}

We have provided information-theoretic lower bounds on Ising model selection with approximate recovery for a variety of graph classes.  For a wide range of scaling regimes of the relevant parameters, we have obtained necessary conditions with the same scaling laws as the best known conditions for exact recovery, thus indicating that approximate recovery is not much easier in the minimax sense.  

To this end, we presented a generalized form of Fano's inequality for handling approximate recovery, and applied it to a variety of graph ensembles.  These were broadly categorized into those where it is difficult to distinguish each graph from the empty graph, and those where it is difficult to determine which edges between highly-correlated groups of nodes are present.  In both cases, we required a departure from the ensembles considered for exact recovery \cite{San12,Sha14} in which the graphs differ in only one or two edges.

It would be interesting to determine to what extent approximate recovery can help when we move beyond the minimax performance criterion and the edit distance.  For example, significant gains may be possible in the setting of random Ising model edge weights $\{\lambda_{ij}\}$, since it may become safe to ``ignore'' the weakest edges.  As another example, since our analysis is based on constructing ensembles of graphs having a KL divergence that is close to a single graph, one may expect that under a recovery criterion based on $D(P_G \| P_{\hat{G}})$ being small, there is more to be gained.  Other directions for further work include models beyond the Ising model (e.g., non-binary, Gaussian), and studies of achieving approximate recovery with \emph{practical} algorithms.

\appendices

\section{Proof of Lemma \ref{lem:Fano_approx}} \label{sec:PF_FANO}

The proof follows standard steps in the derivation of Fano's inequality as in \cite{Sha14}, but with suitable modifications to handle the approximate recovery criterion; see \cite{Ree13} for analogous modifications in the context of support recovery, and \cite{Vat11} for a related list decoding result.  Due to the similarities to other variants, we focus primarily on the details that are specific to the approximate recovery criterion.

Let $G$ be uniformly distributed on $\Tc$, let $\hat{G}$ be the estimate of $G$, and let $E$ and $\hat{E}$ be the corresponding edge sets.  Moreover, let $\pebar(\qmax)$ be the error probability $\PP[|E \Delta E'| > \qmax]$ averaged over the random graph $G$.

By assumption, we may consider decoders such that $\hat{G} \in \Tc'$ without loss of optimality.  Defining the error indicator $\Ec := \openone\{ |E \Delta E'| > \qmax \}$ and applying the chain rule for entropy in two different ways, we have
\begin{align}
    H(\Ec,G|\hat{G}) 
        &= H(G|\hat{G}) + H(\Ec|G,\hat{G}) \label{eq:fano_s1} \\
        &= H(\Ec|\hat{G}) + H(G|\Ec,\hat{G}). \label{eq:fano_s2}
\end{align}
We have $H(\Ec|G,\hat{G}) = 0$ since $\Ec$ is a function of $(G,\hat{G})$, and $H(\Ec|\hat{G}) \le \log 2$ since $\Ec$ is binary.  Moreover, we have
\begin{align}
    & H(G|\Ec,\hat{G}) \nonumber \\
        &\quad = (1-\pebar(\qmax))H(G|\Ec=0,\hat{G}) \nonumber \\
            & \qquad\qquad\qquad\qquad\qquad + \pebar(\qmax) H(G|\Ec=1,\hat{G}) \\
        &\quad \le (1-\pebar(\qmax)) \log A(\qmax) + \pebar(\qmax) \log |\Tc|, \label{eq:fano_s4}
\end{align}
where \eqref{eq:fano_s4} follows from the definition of $A(\dmax)$ in the lemma statement and the fact that $\Ec=0$ implies that $G$ is within a distance $\qmax$ of $\hat{G}$, and we have used the fact that the entropy is upper bounded by the logarithm of the number of elements of the support.  We have now handled three of the terms in \eqref{eq:fano_s1}--\eqref{eq:fano_s2}, and for the final one we write $H(G|\hat{G}) = -I(G;\hat{G}) + H(G) = -I(G;\hat{G}) + \log|\Tc|$, since $G$ is uniform on $\Tc$.

Substituting the preceding observations into \eqref{eq:fano_s1}--\eqref{eq:fano_s2} and performing some simple rearrangements gives
\begin{equation}
    \pebar(\qmax) \ge 1 - \frac{I(G;\hat{G}) + \log 2}{\log|\Tc| - \log A(\qmax)}. \label{eq:fano_s5}
\end{equation}
Finally, we bound the mutual information using the steps of \cite{Sha14}, which are stated here without the details in order to avoid repetition:  We use the data processing inequality to write $I(G;\hat{G}) \le I(G;\Xv)$, where $\Xv$ contains the $n$ independent samples from $P_G$. Using a covering argument, as well as the assumption containing $G'$ in the lemma statement, it follows that $I(G;\Xv) \le n\epsilon$.  Substituting into \eqref{eq:fano_s5}, solving for $n$, and writing $\pe(\qmax) \ge \pebar(\qmax)$, we obtain the desired result.

\bibliographystyle{IEEEtran}
\bibliography{../JS_References}

\begin{thebibliography}{10}
\providecommand{\url}[1]{#1}
\csname url@samestyle\endcsname
\providecommand{\newblock}{\relax}
\providecommand{\bibinfo}[2]{#2}
\providecommand{\BIBentrySTDinterwordspacing}{\spaceskip=0pt\relax}
\providecommand{\BIBentryALTinterwordstretchfactor}{4}
\providecommand{\BIBentryALTinterwordspacing}{\spaceskip=\fontdimen2\font plus
\BIBentryALTinterwordstretchfactor\fontdimen3\font minus
  \fontdimen4\font\relax}
\providecommand{\BIBforeignlanguage}[2]{{%
\expandafter\ifx\csname l@#1\endcsname\relax
\typeout{** WARNING: IEEEtran.bst: No hyphenation pattern has been}%
\typeout{** loaded for the language `#1'. Using the pattern for}%
\typeout{** the default language instead.}%
\else
\language=\csname l@#1\endcsname
\fi
#2}}
\providecommand{\BIBdecl}{\relax}
\BIBdecl

\bibitem{Gem84}
S.~Geman and D.~Geman, ``Stochastic relaxation, {G}ibbs distributions, and the
  {B}ayesian restoration of images,'' \emph{IEEE Trans. Patt. Analysis and
  Mach. Intel.}, no.~6, pp. 721--741, 1984.

\bibitem{Gla63}
R.~J. Glauber, ``Time-dependent statistics of the {I}sing model,'' \emph{J.
  Math. Phys.}, vol.~4, no.~2, pp. 294--307, 1963.

\bibitem{Dur98}
R.~Durbin, S.~R. Eddy, A.~Krogh, and G.~Mitchison, \emph{Biological sequence
  analysis: Probabilistic models of proteins and nucleic acids}.\hskip 1em plus
  0.5em minus 0.4em\relax Cambridge Univ. Press, 1998.

\bibitem{Man99}
C.~D. Manning and H.~Sch{\"u}tze, \emph{Foundations of statistical natural
  language processing}.\hskip 1em plus 0.5em minus 0.4em\relax MIT press, 1999.

\bibitem{Was94}
S.~Wasserman and K.~Faust, \emph{Social network analysis: Methods and
  applications}.\hskip 1em plus 0.5em minus 0.4em\relax Cambridge Univ. Press,
  1994, vol.~8.

\bibitem{Chi96}
D.~M. Chickering, ``Learning {B}ayesian networks is {NP}-complete,'' in
  \emph{Learning from data}.\hskip 1em plus 0.5em minus 0.4em\relax Springer,
  1996, pp. 121--130.

\bibitem{Ree12}
G.~Reeves and M.~Gastpar, ``The sampling rate-distortion tradeoff for sparsity
  pattern recovery in compressed sensing,'' \emph{IEEE Trans. Inf. Theory},
  vol.~58, no.~5, pp. 3065--3092, May 2012.

\bibitem{Sca15b}
J.~Scarlett and V.~Cevher, ``Phase transitions in group testing,'' in
  \emph{Proc. ACM-SIAM Symp. Disc. Alg. (SODA)}, 2016.

\bibitem{San12}
N.~Santhanam and M.~Wainwright, ``Information-theoretic limits of selecting
  binary graphical models in high dimensions,'' \emph{IEEE Trans. Inf. Theory},
  vol.~58, no.~7, pp. 4117--4134, July 2012.

\bibitem{Sha14}
K.~Shanmugam, R.~Tandon, A.~Dimakis, and P.~Ravikumar, ``On the information
  theoretic limits of learning {I}sing models,'' in \emph{Adv. Neur. Inf. Proc.
  Sys. (NIPS)}, 2014.

\bibitem{Isi25}
E.~Ising, ``Beitrag zur theorie des ferromagnetismus,'' \emph{Zeitschrift
  f{\"u}r Physik A Hadrons and Nuclei}, vol.~31, no.~1, pp. 253--258, 1925.

\bibitem{Ana12}
A.~Anandkumar, V.~Y.~F. Tan, F.~Huang, and A.~S. Willsky, ``High-dimensional
  structure estimation in {I}sing models: Local separation criterion,''
  \emph{Ann. Stats.}, vol.~40, no.~3, pp. 1346--1375, 2012.

\bibitem{Ana12a}
------, ``High-dimensional {G}aussian graphical model selection: Walk
  summability and local separation criterion,'' \emph{J. Mach. Learn. Res.},
  vol.~13, pp. 2293--2337, 2012.

\bibitem{Bre08}
G.~Bresler, E.~Mossel, and A.~Sly,
  ``\BIBforeignlanguage{English}{Reconstruction of {M}arkov random fields from
  samples: Some observations and algorithms},'' in
  \emph{\BIBforeignlanguage{English}{Appr., Rand. and Comb. Opt. Algorithms and
  Techniques}}.\hskip 1em plus 0.5em minus 0.4em\relax Springer Berlin
  Heidelberg, 2008, pp. 343--356.

\bibitem{Wu13}
R.~Wu, R.~Srikant, and J.~Ni, ``Learning loosely connected {M}arkov random
  fields,'' \emph{Stoch. Sys.}, vol.~3, no.~2, pp. 362--404, 2013.

\bibitem{Jal11}
A.~Jalali, C.~C. Johnson, and P.~K. Ravikumar, ``On learning discrete graphical
  models using greedy methods,'' in \emph{Adv. Neur. Inf. Proc. Sys. (NIPS)},
  2011.

\bibitem{Ray12}
A.~Ray, S.~Sanghavi, and S.~Shakkottai, ``Greedy learning of graphical models
  with small girth,'' in \emph{Allteron Conf. Comm., Control, and Comp.}, 2012.

\bibitem{Bre14}
G.~Bresler, D.~Gamarnik, and D.~Shah, ``Structure learning of antiferromagnetic
  {I}sing models,'' in \emph{Adv. Neur. Inf. Proc. Sys. (NIPS)}, 2014.

\bibitem{Bre14a}
G.~Bresler, ``Efficiently learning {I}sing models on arbitrary graphs,'' in
  \emph{ACM Symp. Theory Comp. (STOC)}, 2015.

\bibitem{Rav10}
P.~Ravikumar, M.~J. Wainwright, J.~D. Lafferty, and B.~Yu, ``High-dimensional
  {I}sing model selection using $\ell_1$-regularized logistic regression,''
  \emph{Ann. Stats.}, vol.~38, no.~3, pp. 1287--1319, 2010.

\bibitem{Yan14}
E.~Yang, A.~C. Lozano, and P.~K. Ravikumar, ``Elementary estimators for
  graphical models,'' in \emph{Adv. Neur. Inf. Proc. Sys. (NIPS)}, 2014, pp.
  2159--2167.

\bibitem{Mon09}
A.~Montanari and J.~A. Pereira, ``Which graphical models are difficult to
  learn?'' in \emph{Adv. Neur. Inf. Proc. Sys. (NIPS)}, 2009.

\bibitem{Tan13}
R.~Tandon and P.~Ravikumar, ``On the difficulty of learning power law graphical
  models,'' in \emph{IEEE Int. Symp. Inf. Theory}, 2013.

\bibitem{Das12}
A.~K. Das, P.~Netrapalli, S.~Sanghavi, and S.~Vishwanath, ``Learning {M}arkov
  graphs up to edit distance,'' in \emph{IEEE Int. Symp. Inf. Theory}, 2012,
  pp. 2731--2735.

\bibitem{Vat11}
D.~Vats and J.~M. Moura, ``Necessary conditions for consistent set-based
  graphical model selection,'' in \emph{IEEE Int. Symp. Inf. Theory}, 2011, pp.
  303--307.

\bibitem{Mei06}
N.~Meinshausen and P.~B\"uhlmann, ``High-dimensional graphs and variable
  selection with the {L}asso,'' \emph{Ann. Stats.}, vol.~34, no.~3, pp.
  1436--1462, June 2006.

\bibitem{Wan10}
W.~Wang, M.~Wainwright, and K.~Ramchandran, ``Information-theoretic bounds on
  model selection for {G}aussian {M}arkov random fields,'' in \emph{IEEE Int.
  Symp. Inf. Theory}, 2010.

\bibitem{Jog15}
V.~Jog and P.-L. Loh, ``On model misspecification and {KL} separation for
  {G}aussian graphical models,'' in \emph{IEEE Int. Symp. Inf. Theory}, 2015.

\bibitem{Rav11}
P.~Ravikumar, M.~J. Wainwright, G.~Raskutti, and B.~Yu, ``High-dimensional
  covariance estimation by minimizing $\ell_1$-penalized log-determinant
  divergence,'' \emph{Elec. J. Stats.}, vol.~5, pp. 935--980, 2011.

\bibitem{Cov01}
T.~M. Cover and J.~A. Thomas, \emph{Elements of Information Theory}.\hskip 1em
  plus 0.5em minus 0.4em\relax John Wiley \& Sons, Inc., 2006.

\bibitem{Ree13}
G.~Reeves and M.~Gastpar, ``Approximate sparsity pattern recovery:
  Information-theoretic lower bounds,'' \emph{IEEE Trans. Inf. Theory},
  vol.~59, no.~6, pp. 3451--3465, June 2013.

\end{thebibliography}

\begin{IEEEbiographynophoto}{Jonathan Scarlett}
(S'14 -- M'15) received 
the B.Eng. degree in electrical engineering and the B.Sci. degree in 
computer science from the University of Melbourne, Australia. In 2011, 
he was a research assistant at the Department of Electrical \& Electronic 
Engineering, University of Melbourne.  From October 2011 to August 2014,  
he was a Ph.D. student in the Signal Processing and Communications Group
at the University of Cambridge, United Kingdom. He
is now a post-doctoral researcher with the Laboratory for Information
and Inference Systems at the \'Ecole Polytechnique F\'ed\'erale de Lausanne,
Switzerland.  His research interests are in the areas of information theory, 
signal processing, machine learning, and high-dimensional statistics. 
He received the Cambridge Australia Poynton International Scholarship, and the EPFL Fellows postdoctoral fellowship co-funded by Marie Curie.
\end{IEEEbiographynophoto}

\begin{IEEEbiographynophoto}{Volkan Cevher}
(SM'10) received the B.Sc. (valedictorian)
in electrical engineering from Bilkent
University in Ankara, Turkey, in 1999 and the Ph.D.
in electrical and computer engineering from the
Georgia Institute of Technology in Atlanta, GA in
2005. He was a Research Scientist with the University
of Maryland, College Park from 2006-2007 and also
with Rice University in Houston, TX, from 2008-2009.
Currently, he is an Associate Professor at the
Swiss Federal Institute of Technology Lausanne and
a Faculty Fellow in the Electrical and Computer
Engineering Department at Rice University. His research interests include signal
processing theory, machine learning, convex optimization, and information
theory. Dr. Cevher was the recipient of a Best Paper Award at SPARS in 2009,
a Best Paper Award at CAMSAP in 2015, and an ERC StG in 2011.
\end{IEEEbiographynophoto}

\end{document}